\documentclass[11pt]{article}

\usepackage[a4paper,margin=30mm]{geometry}
\usepackage{amsmath,amssymb,amsthm,physics,comment,xcolor}
\usepackage{graphicx}
\usepackage{cancel}
\usepackage{cite}
\usepackage{hyperref}
\usepackage{ulem}

\theoremstyle{plain}
\newtheorem{Thm}{Theorem}
\newtheorem{Prop}{Proposition}
\newtheorem{Lem}{Lemma}
\newtheorem{Cor}{Corollary}
\newtheorem{Def}{Definition}
\newtheorem{Rem}{Remark}

\newtheorem{F}{Fact}

\def\<{\mathopen{(\!(}}
\def\>{\mathclose{)\!)}}

\newcommand{\CA}{\mathbb{C}}

\newcommand{\id}{\mathrm{id}}
\newcommand{\II}{\mathop{\mathbb{I}}\nolimits}

\newcommand{\LA}{\mathop{\mathcal{L}}\nolimits}
\newcommand{\HA}{\mathop{\mathcal{H}}\nolimits}
\newcommand{\SA}{\mathop{\mathcal{S}}\nolimits}

\newcommand{\I}{\mathop{\mathrm{id}}\nolimits}

\newcommand{\R}{\mathop{\mathbb{R}}\nolimits}
\newcommand{\N}{\mathop{\mathbb{N}}\nolimits}

\title{
Another Legacy of Andrzej Kossakowski:\thanks{Dedicated to the memory of Andrzej Kossakowski.}\\
A Self-Contained Derivation of the GKLS Equation
}
\author{
Gen Kimura\thanks{Email: \texttt{gen.kimura.quant@gmail.com}}\\
\small Graduate School of Information Sciences, Tohoku University\\
\small Sendai, Japan
}

\date{}

\begin{document}

\maketitle

\begin{abstract}
This note is written for the special issue of OSID dedicated to the
50th anniversary of the Gorini--Kossakowski--Lindblad--Sudarshan equation.
Its purpose is not to give a comprehensive historical review, but rather to
reconstruct, in a self-contained way, one logical route leading to the GKLS
generator. The emphasis is placed on Kossakowski's structural insight: the
combination of Markovianity, complete positivity, trace preservation, and
the infinitesimal structure of quantum dynamical semigroups naturally leads
to the celebrated form of the generator. 
I also include a few personal recollections of Andrzej Kossakowski, under whose guidance I spent one year as a postdoctoral researcher in 2003--2004.
\end{abstract}

\section{Introduction}
The purpose of this note is to reconstruct a logically closed route to the GKLS (Gorini--Kossakowski--Lindblad--Sudarshan) \cite{GoriniKossakowskiSudarshan1976,Lindblad1976} representation theorem for the generator of quantum Markovian dynamics, understood here as completely positive dynamical semigroups, in finite-level systems. 
The route we follow is mainly centered on Kossakowski's structural insight,
which was later developed into the formulation of the GKS theorem \cite{GoriniKossakowskiSudarshan1976}.
Already in his earlier work on positive dynamical semigroups, he had developed the axiomatic framework of quantum dynamical semigroups and studied their generator conditions in the early 1970s\cite{Kossakowski1972a,Kossakowski1972b,Kossakowski1973}. 
The now-standard finite-dimensional $N$-level form of the generator, expressed in terms of a traceless Hilbert--Schmidt orthonormal basis, was established by Gorini, Kossakowski, and Sudarshan \cite{GoriniKossakowskiSudarshan1976}. 
Independently, Lindblad obtained the equivalent theorem in a more general operator-algebraic setting, which partially extends the finite-dimensional picture to infinite-dimensional settings~\cite{Lindblad1976}. 
From this viewpoint, once the importance of complete positivity in open quantum system dynamics was recognized in the physics community, the passage to the GKLS form appears, at least in retrospect, to have been only a matter of time, building on the structural ideas that had already been developed in the early work of Kossakowski and other pioneers\footnote{For a historical account of the GKLS equation, see \cite{ChruscinskiPascazio2017}.}. 
In this note, we give a logically self-contained account of the passage from Kossakowski's characterization of positive dynamical semigroups (Theorem~\ref{thm:K}) to the GKS representation theorem (Theorem~\ref{thm:GKS}) for their generators\footnote{Although we do not follow Lindblad's independent proof in this note, his original paper \cite{Lindblad1976} is well worth reading for its entirely independent and highly original ideas.}. 

This note also includes some personal memories and anecdotes from the year I
spent with Professor Kossakowski, whom I had the privilege of knowing as a mentor.\footnote{
In this note, I use the first person singular when speaking about my personal recollections of Prof. Kossakowski.
In the mathematical exposition, I use the conventional inclusive ``we''.}
In writing this note, and in recalling his familiar smile, I would like to pass on to the reader, however modestly, something of what I was fortunate enough to receive from him. \footnote{I also warmly invite the reader to read Saverio Pascazio's memorial note~\cite{Saverio}, ``Remembering Andrzej Kossakowski''.}

Indeed, the choice of a self-contained presentation reflects a piece of advice he once gave me. He told me that a paper should be made as logically closed as
possible. Behind this advice was a simple and generous concern for the reader:
if substantial parts of the argument are left only to references, then readers who cannot easily obtain those references may be unable to understand the work.
Since then, I have tried to follow this advice in my own papers, and it is in this spirit that I have tried to make the present note as self-contained as possible.

Finally, the title of this note is also intended as a homage to Dariusz Chru\'sci\'nski's article \cite{Chruscinski2021LegacyKossakowski}, ``The Legacy of Andrzej Kossakowski.''
While that article beautifully presents Kossakowski's scientific legacy, I would like to describe here another aspect of that legacy from my own personal perspective.

\bigskip 

The structure of this note is as follows.
In Sec.~\ref{sec:PrNot}, we introduce the notation and review some basic results that will be used freely throughout the note.
These results are standard facts frequently used in the theory of open quantum systems and quantum information theory.
Readers familiar with this material may consult the notation as needed and
proceed directly to the next section. 
In Sec.~\ref{sec:OS}, we discuss positive and trace-preserving dynamical maps, their extension to complete positivity, and Markovianity formulated in terms of the semigroup law.

\section{Preliminaries and notation}\label{sec:PrNot}
The prerequisites for the present note are some familiarity with the foundations of quantum mechanics and quantum information theory for finite-level systems, as well as basic linear algebra and calculus.
We refer the reader to standard textbooks on quantum information science, e.g.,
\cite{NC,OurTextbook,Wilde,Watrous,HeinosaariZiman}, and on the theory of open quantum systems, e.g., \cite{BreuerPetruccione2002,RivasHuelga2012,AlickiLendi2007}.

To make the present note as logically self-contained and pedagogical as possible, we collect some basic results needed in the discussion as Facts.
These Facts may be used without further mention; they are elementary results whose verification may help readers become familiar with the basic tools used in this note.
For the reader's convenience, however, proofs of them are included in \ref{app:Facts}.

One point worth mentioning, in connection with the self-contained nature of this note, concerns some elementary but technical facts about topology and limiting procedures. 
In finite dimensions, all norms on a vector space are equivalent, and consequently distinctions such as strong convergence and operator-norm convergence do not play an essential role; see, e.g., \cite{RudinPMA,Kreyszig}.
Indeed, for readers with a physics background, these limiting arguments may be regarded as the familiar finite-dimensional manipulations usually treated intuitively.
For readers with a mathematical background, these facts will likely be familiar.
Therefore, we do not expect this point to obscure the main line of the argument.

\subsection{Mathematical preliminaries}\label{sec:MP}
For a $d$-dimensional Hilbert space $\HA \simeq \CA^d$, we write $\braket{\psi}{\phi}$ for the inner product of $\psi,\phi \in \HA$, using the physics convention that it is conjugate-linear in the first argument and linear in the second argument:
\[
\braket{\alpha\psi+\beta\psi'}{\phi}
=
\overline{\alpha}\braket{\psi}{\phi}
+
\overline{\beta}\braket{\psi'}{\phi},
\qquad
\braket{\psi}{\alpha \phi+\beta \phi'}
=
\alpha\braket{\psi}{\phi}
+
\beta\braket{\psi}{\phi'}.
\]
Here $\alpha,\beta \in \CA$ and $\psi,\psi',\phi,\phi' \in \HA$. 
The induced norm is given by $\|\psi\|:=\sqrt{\braket{\psi}{\psi}}$.
A family $\{\phi_i\}_{i=1}^n$ of vectors in $\HA$ is called an orthonormal set if $\braket{\phi_i}{\phi_j}=\delta_{ij}$ for all $i,j=1,\ldots,n$, where $\delta_{ij}$ is the Kronecker delta, i.e., $\delta_{ij}=1$ if $i=j$ and $\delta_{ij}=0$ otherwise.
If $n=d$, the dimension of $\HA$, such a family is called an orthonormal basis (ONB). 

Let $\LA(\HA) \simeq M_d(\CA)$ be the set of all linear operators on $\HA$. 
For $\psi,\phi \in \HA$, $\ketbra{\psi}{\phi}$ denotes the linear operator on $\HA$ whose action on $\xi \in \HA$ is given by $\ketbra{\psi}{\phi}\xi := \braket{\phi}{\xi}\psi$. 
An orthonormal basis $\{\phi_i\}_{i=1}^d$ satisfies the completeness relation $\sum_{i=1}^d \ketbra{\phi_i}{\phi_i}=\II$.

For $A \in \LA(\HA)$, its trace is defined by
\begin{equation}
\Tr A := \sum_{i=1}^d \langle \phi_i | A\phi_i \rangle,\label{eq:tr}
\end{equation}
where $\{\phi_i\}_{i=1}^d$ is an arbitrary orthonormal basis of $\HA$\footnote{
Note that this definition is independent of the choice of the ONB.
Indeed, let $\{\psi_j\}_{j=1}^d$ be another ONB.
Using the completeness relations $\sum_i \ketbra{\phi_i}{\phi_i}=\II$ and $\sum_j \ketbra{\psi_j}{\psi_j}=\II$, we have
\[
\sum_i \langle \phi_i|A\phi_i\rangle
=
\sum_{i,j}\langle \phi_i|A\psi_j\rangle\langle \psi_j|\phi_i\rangle
=
\sum_{i,j}\langle \psi_j|\phi_i\rangle\langle \phi_i|A\psi_j\rangle
=
\sum_j \langle \psi_j|A\psi_j\rangle.
\]
}.
The cyclic property of the trace, $\Tr AB=\Tr BA$, is particularly useful and will be used throughout this note.
We shall also use the identity
\[
\Tr(\ketbra{\psi}{\phi}A)=\braket{\phi}{A\psi}
\]
for $\psi,\phi\in\HA$ and $A\in\LA(\HA)$.
In particular, setting $A=\I$ gives $\Tr(\ketbra{\psi}{\phi})=\braket{\phi}{\psi}.$

For $A\in\LA(\HA)$, its adjoint $A^\dagger$ is defined by $\braket{\psi}{A\phi}=\braket{A^\dagger\psi}{\phi}$ for all $\psi,\phi\in\HA$. Note that, for $A,B\in\LA(\HA)$, one has $(AB)^\dagger=B^\dagger A^\dagger$. 
For the commutator and the anticommutator, we use the notation
$[A,B]:=AB-BA$ and $\{A,B\}:=AB+BA$, respectively. 
We say that $A$ is normal if $[A,A^\dagger]=0$, self-adjoint (or Hermitian) if $A=A^\dagger$, unitary if $AA^\dagger=A^\dagger A=\II$ (where $\II$ denotes the identity operator), positive\footnote{More precisely, this should be called positive semidefinite, but in this note we simply call it positive.}, written $A \geq 0$, if $\langle \psi | A\psi\rangle \geq 0$ for all $\psi \in \HA$, and a projection if $A=A^\dagger$ and $A^2=A$. 
We write $\LA_0(\HA)$, $\LA_{\rm sa}(\HA)$, and $\LA_+(\HA)$ for the sets of traceless, self-adjoint, and positive operators on $\HA$, respectively. 

One of the fundamental results in linear algebra is that $A$ is normal if and only if it admits an eigenvalue decomposition $A = \sum_{i=1}^d a_i \ketbra{\phi_i}{\phi_i}$, where $a_i$ are the eigenvalues of $A$ and $\{\phi_i\}_{i=1}^d$ is an ONB of
$\HA$ consisting of corresponding eigenvectors. 
Equivalently, $A$ admits the spectral decomposition $A=\sum_{\mu} \alpha_\mu P_\mu$, where $\alpha_\mu$ are the distinct eigenvalues of $A$ and $P_\mu$ is the orthogonal projection onto the eigenspace corresponding to $\alpha_\mu$. 
The eigen projections satisfy $P_\mu P_\nu=\delta_{\mu\nu}P_\mu$ and $\sum_\mu P_\mu=\II$. 
For normal operators, the following basic spectral characterizations hold:
\begin{enumerate}
\item $A$ is unitary if and only if all eigenvalues of $A$ have modulus one.
\item $A$ is self-adjoint if and only if all eigenvalues of $A$ are real.
\item $A$ is positive if and only if all eigenvalues of $A$ are nonnegative.
\item $A$ is a projection if and only if all eigenvalues of $A$ are either $0$ or $1$.
\end{enumerate}

For a normal operator with eigenvalue decomposition $A=\sum_{i=1}^d a_i \ketbra{\phi_i}{\phi_i}$ and a function
$f:\CA\to\CA$, we define $f(A):=\sum_{i=1}^d f(a_i)\ketbra{\phi_i}{\phi_i}$.
In particular, if $A\geq 0$, its positive square root is $\sqrt{A}:=\sum_{i=1}^d \sqrt{a_i}\ketbra{\phi_i}{\phi_i}$.
For non-normal operators, the definition of $f(A)$ requires more care; nevertheless, some elementary functions are defined directly for arbitrary $A\in\LA(\HA)$, such as powers $A^n$, defined recursively by $A^0:=\II$ and $A^{n+1}:=AA^n$ and the exponential $\exp(A):=\sum_{n=0}^{\infty}A^n/n!$, which is well defined in finite dimensions, and the absolute value $|A|:=\sqrt{A^\dagger A}$. 

The space $\LA(\HA)$ is itself a $d^2$-dimensional Hilbert space with respect to the Hilbert--Schmidt inner product $\langle A,B\rangle_{\rm HS}:=\Tr(A^\dagger B)$ for $A,B\in\LA(\HA)$.
The corresponding Hilbert--Schmidt norm is denoted by $\|A\|_{\rm HS}:=\sqrt{\Tr(A^\dagger A)}$.
If $\{\phi_i\}_{i=1}^d$ is an orthonormal basis of $\HA$, then the matrix units $E_{ij}:=\ketbra{\phi_i}{\phi_j}$, $i,j=1,\ldots,d$, form an orthonormal basis of $\LA(\HA)$.

We shall also use the trace norm $\|A\|_1:=\Tr|A|$, with $|A|:=\sqrt{A^\dagger A}$, and the operator norm $\|A\|_\infty:=\sup_{\|\psi\|=1}\|A\psi\|$. 
In finite dimensions, all norms on $\LA(\HA)$ are equivalent; hence the notion of convergence does not depend on the choice of norm. Thus $A=\lim_{n\to\infty}A_n$ means, for any norm $\|\cdot\|_*$, that
$\lim_{n\to\infty}\|A-A_n\|_*=0$.

We use tensor products of finite-dimensional Hilbert spaces.
If $\HA_1$ and $\HA_2$ are finite-dimensional Hilbert spaces of dimensions $d_1$ and $d_2$, respectively, their tensor product is denoted by $\HA_1\otimes\HA_2$ and has dimension $d_1d_2$.
For $\psi\in\HA_1$ and $\phi\in\HA_2$, we write $\psi\otimes\phi$ for the corresponding product vector.
The tensor product is bilinear in the sense that
\[
(\alpha\psi+\beta\psi')\otimes\phi
=
\alpha(\psi\otimes\phi)+\beta(\psi'\otimes\phi),
\qquad
\psi\otimes(\alpha\phi+\beta\phi')
=
\alpha(\psi\otimes\phi)+\beta(\psi\otimes\phi')
\]
for $\alpha,\beta\in\CA$, $\psi,\psi'\in\HA_1$, and $\phi,\phi'\in\HA_2$.
The inner product on $\HA_1\otimes\HA_2$ is defined on product vectors by
\[
\braket{\psi_1\otimes\phi_1}{\psi_2\otimes\phi_2}
=
\braket{\psi_1}{\psi_2}\braket{\phi_1}{\phi_2},
\]
and is extended to all of $\HA_1\otimes\HA_2$ by linearity in the second argument and conjugate-linearity in the first argument.
If $\{e_i\}_{i=1}^{d_1}$ and $\{f_j\}_{j=1}^{d_2}$ are orthonormal bases of $\HA_1$ and $\HA_2$, respectively, then $\{e_i\otimes f_j\}_{i=1,\ldots,d_1,\ j=1,\ldots,d_2}$ is an orthonormal basis of $\HA_1\otimes\HA_2$.

For $A\in\LA(\HA_1)$ and $B\in\LA(\HA_2)$, the tensor product operator $A\otimes B$ is defined by
$(A\otimes B)(\psi\otimes\phi):=A\psi\otimes B\phi$ and extended linearly.
It satisfies
$(A\otimes B)^\dagger=A^\dagger\otimes B^\dagger$,
$(A_1\otimes B_1)(A_2\otimes B_2)=A_1A_2\otimes B_1B_2$,
and
$\Tr(A\otimes B)=\Tr A\Tr B$.
Under the natural identification $\LA(\HA_1)\otimes\LA(\HA_2)\simeq\LA(\HA_1\otimes\HA_2)$, every operator on $\HA_1\otimes\HA_2$ can be written as a finite sum $\sum_k A_k\otimes B_k$.

We shall also use the partial trace.
For $X\in\LA(\HA_1\otimes\HA_2)$, the partial trace over $\HA_2$ is the unique linear map $\Tr_2:\LA(\HA_1\otimes\HA_2)\to\LA(\HA_1)$ characterized by
$\Tr_2(A\otimes B)=A\Tr B$ for all $A\in\LA(\HA_1)$ and $B\in\LA(\HA_2)$.

\subsection{Quantum states and unitary dynamics}

According to quantum theory, every quantum system is associated with a Hilbert space. A quantum system is called a finite-level system, or more specifically a $d$-level system, if its associated Hilbert space is finite-dimensional, $\HA \simeq \CA^d$.
In quantum information science, such a system is sometimes called a {\it qudit} system, and in particular a {\it qubit} system when $d=2$. 
In this note, the quantum system of interest is always assumed to be a finite-level system. 
A quantum state is represented by a density operator $\rho \in \LA(\HA)$, namely a positive operator $\rho \geq 0$ with unit trace, $\Tr \rho=1$. 
The set $\SA(\HA)$ of density operators is convex: for $\rho_1,\rho_2 \in \SA(\HA)$ and $p \in [0,1]$, the state
$\rho=p\rho_1+(1-p)\rho_2$ represents the probabilistic mixture obtained by preparing $\rho_1$ with probability $p$ and $\rho_2$ with probability $1-p$.

For an isolated quantum system, time evolution is unitary: a state $\rho$ evolves as
$\rho \mapsto \rho' = U\rho U^\dagger$. Introducing a time parameter $t \in \mathbb{R}$, the time evolution from an initial state $\rho$ to the state $\rho_t$ at time $t$ is given by
$\rho_t=U_t\rho U_t^\dagger$, where $U_t$ is a unitary operator. Equivalently, $\rho_t$ satisfies the Schr\"odinger--von Neumann equation
\begin{equation}
    \frac{d\rho_t}{dt}=-i[H,\rho_t],
\end{equation}
where $H=H^\dagger$ is the (time-independent) Hamiltonian of the system and $U_t=\exp(-iHt)$.\footnote{Throughout this note, we set the reduced Planck constant $\hbar$ equal to $1$.}

However, an open quantum system cannot in general be described by unitary evolution alone, since phenomena such as decoherence, dissipation, and thermalization arise through its interaction with the surrounding environment. 
The main subject of this note is the mathematical description of such non-unitary time evolutions, especially those forming quantum Markovian dynamics characterized by a semigroup law.

\section{Quantum Markovian Dynamics}\label{sec:OS}

\subsection{Mathematical preliminaries for open quantum systems}

Let $\HA \simeq \CA^d$ be a finite-dimensional Hilbert space. 
A linear map $\Lambda$ on $\LA(\HA)$, or a real-linear map $\Lambda$ on $\LA_{\rm sa}(\HA)$, is called trace-preserving if $\Tr \Lambda(A)=\Tr A$ for all $A$ in the corresponding space. 
It is called positive\footnote{More explicitly, this means positivity-preserving, but the standard terminology is simply positive.} if $\Lambda(A)\geq 0$ for all $A\geq 0$.

\begin{Rem}\label{rem:LAsaLA}
Real-linear maps on $\LA_{\rm sa}(\HA)$ can always be identified unambiguously with their complex-linear extensions to $\LA(\HA)$. Indeed, every operator $A\in\LA(\HA)$ admits the Cartesian decomposition
\begin{equation}\label{CDec}
A=A_R+iA_I,
\qquad
A_R:=\frac{A+A^\dagger}{2},
\qquad
A_I:=\frac{A-A^\dagger}{2i},
\end{equation}
where $A_R,A_I\in\LA_{\rm sa}(\HA)$. Therefore, a real-linear map on
$\LA_{\rm sa}(\HA)$ uniquely determines a complex-linear extension to
$\LA(\HA)$ by
\[
\Lambda(A):=\Lambda(A_R)+i\Lambda(A_I).
\]
In what follows, we use the same symbol for a real-linear map on
$\LA_{\rm sa}(\HA)$ and for its complex-linear extension to $\LA(\HA)$,
according to the context.
\end{Rem}

One can extend the notion of positivity as follows. For $n\in\N$, a linear map
$\Lambda$ on $\LA(\HA)$ is called $n$-positive if the map
\[
\Lambda\otimes\I_n:
\LA(\HA)\otimes M_n(\CA)\simeq \LA(\HA\otimes\CA^n)
\longrightarrow
\LA(\HA)\otimes M_n(\CA)\simeq \LA(\HA\otimes\CA^n)
\]
is positive, where $\I_n$ denotes the identity map on $M_n(\CA)$. 
Equivalently, $\Lambda\otimes\I_n$ is required to preserve the positivity of operators on the composite system $\HA\otimes\CA^n$. 
The map $\Lambda$ is called completely positive \cite{Stinespring1955}, or CP, if it is $n$-positive for all $n\in\N$.
Note that $1$-positivity is equivalent to positivity. 
If $\Lambda$ is $n$-positive, then it is $m$-positive for every $m\leq n$.

We recall the following standard characterizations of completely positive maps \cite{Choi1975,Jamiolkowski1972,Kraus1971,Stinespring1955}.
\begin{Thm}\label{thm:cp}
For a linear map $\Lambda$ on $\LA(\HA)$ with $\HA \simeq \CA^d$, the following conditions are equivalent:
\begin{enumerate}
\item[(i)] $\Lambda$ is completely positive.
\item[(ii)] $\Lambda$ is $d$-positive.
\item[(iii)] For a maximally entangled state $\Psi_M:=\frac{1}{\sqrt{d}}\sum_{i=1}^d e_i\otimes e_i$, where $\{e_i\}_{i=1}^d$ is an orthonormal basis of $\HA$, one has
\[
(\Lambda\otimes\id_d)(\ketbra{\Psi_M}{\Psi_M})\geq 0.
\]
\end{enumerate}
\end{Thm}

\begin{Rem}
By definition, (i) $\Rightarrow$ (ii) $\Rightarrow$ (iii) are immediate.
The theorem says that the converse implications also hold.
In particular, for a linear map on $\LA(\HA)$ with $\dim\HA=d$, complete positivity is already equivalent to $d$-positivity.
Moreover, although $d$-positivity of $\Lambda$ means that $\Lambda\otimes\id_d$ preserves positivity for all positive operators on $\HA\otimes\CA^d$, for complete positivity it is enough to check positivity only for the single maximally entangled state $\ketbra{\Psi_M}{\Psi_M}$.
The operator $C_\Lambda:=(\Lambda\otimes\id_d)(\ketbra{\Psi_M}{\Psi_M})$ is called the Choi matrix of $\Lambda$, and the correspondence $\Lambda \mapsto C_\Lambda$ is known as the Choi--Jamio{\l}kowski isomorphism.
\end{Rem}

The following theorem gives the Choi--Kraus representation of completely positive maps \cite{Kraus1971,Choi1975,Stinespring1955}.
\begin{Thm}\label{thm:cp2}
For a linear map $\Lambda$ on $\LA(\HA)$ with $\HA \simeq \CA^d$, the following conditions are equivalent:
\begin{enumerate}
\item[(i)] $\Lambda$ is completely positive.
\item[(ii)] There exist operators $M_\alpha\in\LA(\HA)$ such that
\[
\Lambda(A)
=
\sum_\alpha M_\alpha A M_\alpha^\dagger
\qquad
(A\in\LA(\HA)).
\]
\end{enumerate}
Moreover, $\Lambda$ is trace-preserving if and only if the operators $M_\alpha$ can be chosen so that
\[
\sum_\alpha M_\alpha^\dagger M_\alpha=\II.
\]
\end{Thm}
\noindent For completeness, we give proofs of Theorems~\ref{thm:cp} and \ref{thm:cp2} in \ref{app:cp}.

We denote by $\LA(\LA(\HA))$ the set of all linear maps on $\LA(\HA)$.
Since $\LA(\HA)$ is itself a Hilbert space with respect to the Hilbert--Schmidt inner product, we can also regard $\LA(\LA(\HA))$ as the space of linear operators on this Hilbert space. Hence $\LA(\LA(\HA))$ carries its own Hilbert--Schmidt inner product: 
For $\Gamma,\Phi\in\LA(\LA(\HA))$, we define
\[
    \langle \Gamma,\Phi\rangle_{\rm HS}
    :=
    \Tr(\Gamma^\dagger\Phi),
\]
where $\Gamma^\dagger$ denotes the adjoint of $\Gamma$ with respect to the
Hilbert--Schmidt inner product on $\LA(\HA)$, i.e.,
$\braket{A}{\Gamma(B)}_{\rm HS}
=
\braket{\Gamma^\dagger(A)}{B}_{\rm HS}$ for all $A,B\in\LA(\HA)$, and the
trace is taken over the Hilbert space $\LA(\HA)$.
Equivalently, if
$\{G_\alpha\}_{\alpha=1,\ldots,d^2}$ is any Hilbert--Schmidt orthonormal
basis of $\LA(\HA)$, then
\begin{equation}\label{eq:HSIL}
    \langle \Gamma,\Phi\rangle_{\rm HS}=
    \sum_\alpha
    \langle G_\alpha,\Gamma^\dagger\Phi(G_\alpha)\rangle_{\rm HS} =
    \sum_\alpha
    \langle \Gamma(G_\alpha),\Phi(G_\alpha)\rangle_{\rm HS} =
    \sum_\alpha
    \Tr\left[(\Gamma(G_\alpha))^\dagger\Phi(G_\alpha)\right].
\end{equation}
The corresponding Hilbert--Schmidt norm is given by
$\|\Gamma\|_{\rm HS}=\sqrt{\langle \Gamma,\Gamma\rangle_{\rm HS}}$. Besides the Hilbert--Schmidt norm, one can introduce various other norms on spaces of linear maps. 
In particular, we mainly use the following operator norm induced by the trace norm.
For a real-linear map $\Lambda$ on $\LA_{\rm sa}(\HA)$, we define
\begin{equation}\label{eq:OpNormLam}
   \|\Lambda\|_{1 \to 1}
   :=
   \sup_{\substack{A\in\LA_{\rm sa}(\HA)\\ A \neq 0 }}
   \frac{\|\Lambda(A)\|_1}{\|A\|_1}.
\end{equation}

\subsection{Quantum dynamics from an operational viewpoint}

In this subsection, assuming that quantum states are represented by density operators, we characterize quantum dynamics by imposing natural operational requirements. Let the system of interest be a finite-level system\footnote{In open quantum systems, the system of interest can often be genuinely finite-dimensional, as exemplified by an electronic spin system or the polarization degree of freedom of a photon.
The environment may still be infinite-dimensional, or may possess infinitely many degrees of freedom; the finite-dimensional restriction concerns only the quantum system under consideration.} $S$ with Hilbert space $\HA \simeq \CA^d$. 
Suppose that the transformation of the system from an initial state to a final state is described by a time evolution map $\Lambda$. 
Since a time evolution map should naturally send states to states, $\Lambda$ must at least map the state space $\SA(\HA)$ into itself. 
Moreover, since time evolution should also preserve probabilistic mixtures of preparations, $\Lambda$ should be affine on $\SA(\HA)$; that is,
\[
    \Lambda(p\rho_1+(1-p)\rho_2)
    =
    p\Lambda(\rho_1)+(1-p)\Lambda(\rho_2)
\]
for all $\rho_1,\rho_2 \in \SA(\HA)$ and $p \in [0,1]$. 
\begin{F}\label{F:LExt}
It is a standard fact that every affine map $\Lambda:\SA(\HA)\to\SA(\HA)$ uniquely extends to a trace-preserving positive real-linear map on $\LA_{\rm sa}(\HA)$. 
\end{F}
Moreover, as explained in Remark~\ref{rem:LAsaLA}, $\Lambda$ can be further
extended to $\LA(\HA)$. Thus, whether for an isolated or an open quantum system,
a time evolution map may be naturally viewed as a trace-preserving positive map
on $\LA_{\rm sa}(\HA)$, or equivalently through its complex-linear extension on
$\LA(\HA)$.

However, in the dynamics of open quantum systems, it is customary to impose the stronger condition of complete positivity rather than mere positivity. 
One simple motivation for this stronger requirement is the following. 
The state of the system of interest may be only part of a larger composite system, namely a reduced density operator obtained by ignoring an ancillary or environmental system. Therefore, even when the dynamics acts only on the system of interest, it should preserve the positivity of states of the composite system. 
More precisely, if an arbitrary ancillary system $\CA^n$ is attached to the system and left unaffected by the dynamics, the extended map is naturally given by $\Lambda\otimes\I_n$. Complete positivity is precisely the requirement that $\Lambda\otimes\I_n$ preserves positivity for every $n\in\N$. Thus complete positivity guarantees physical consistency even when the system is correlated or entangled with an external system. Since trace preservation of $\Lambda\otimes\I_n$ follows from that of $\Lambda$, no additional trace condition is needed. 

Another convincing reason\cite{GoriniKossakowskiSudarshan1976} comes from the characterization of time evolution
maps in terms of reduced dynamics. Suppose that the system $S$ interacts with
an environment $E$. If the total system $S+E$ is regarded as an isolated system,
then its dynamics is described by a unitary evolution on $S+E$, namely
$X\mapsto UXU^\dagger$. If the system and the environment are initially
uncorrelated, then, for an initial state $\rho$ of the system and an initial
state $\rho_E$ of the environment, often taken to be a thermal equilibrium state,
the initial state of the total system is $\rho\otimes\rho_E$, and it evolves as
$ \rho\otimes\rho_E \mapsto U(\rho\otimes\rho_E)U^\dagger$. 
The resulting state of the system is obtained by tracing out the environment:
\[
\Lambda(\rho)
=
\Tr_E\left[U(\rho\otimes\rho_E)U^\dagger\right].
\]
It is straightforward to verify that this construction gives a completely
positive map\footnote{
Indeed, positivity is clear because $\rho\otimes\rho_E\geq 0$, unitary conjugation preserves positivity, and the partial trace over $E$ maps positive operators to positive operators.
Complete positivity follows by applying the same argument after adjoining an arbitrary $n$-dimensional ancillary system:
\begin{equation}
(\Lambda\otimes \I_n)(X)
=
\Tr_E\left[
(U\otimes I_n)(X\otimes\rho_E)(U^\dagger\otimes I_n)
\right].
\end{equation}
}.
Consequently, a time evolution map is generally assumed to be a completely positive and trace-preserving (CPTP) map. In quantum information theory, such a map is also called a quantum channel.

\begin{Rem}
Whether complete positivity should be regarded as a universal requirement for all quantum dynamics has often been discussed.
In our own work, we have pursued this question from the viewpoint of experimentally testable universal constraints on relaxation times implied by the GKLS equation.
This line of research began with the qubit case~\cite{Rel1}.
For general $d$-level systems, the tightest form of such a constraint was conjectured in~\cite{Rel2} and proved there for generic physical systems.
I would like to record here, with deep respect, that this was Kossakowski's last paper.
The conjecture was later resolved in full in~\cite{Rel3}, and the same viewpoint has subsequently developed into the study of experimentally testable consequences of $n$-positivity~\cite{Rel4}.
\end{Rem}

To describe the time evolution continuously, we introduce a time parameter and consider a family of time evolution maps $\{\Lambda_t\}_{t\geq 0}$. 
From a physical point of view, it is natural to assume that the state of the system changes continuously in time. Since the trace norm has an operational meaning, for instance in terms of the optimal distinguishability of quantum states, it is also natural to use the trace norm to measure the distance between states.
Accordingly, the continuity of the time evolution may be formulated as
\begin{equation}\label{eq:TimeConti}
\|\Lambda_t(\rho)-\Lambda_s(\rho)\|_1 \to 0 \quad (t\to s)
\end{equation}
for every state $\rho$.

So far, we have discussed the general theory of time evolution maps for open quantum systems. We now turn to an important class of dynamics, namely Markovian dynamics.
In many open quantum systems, Markovian dynamics provides an effective description of dissipative processes such as spin relaxation in NMR, spontaneous emission and damping in quantum optics, and thermal relaxation of a small quantum system weakly coupled to a large reservoir. In such cases, the environmental correlation time is much shorter than the characteristic time scale of the system, so that memory effects can be neglected \cite{GardinerZoller2004,BreuerPetruccione2002}.

Although it is difficult to give a completely unambiguous definition of Markovianity in open quantum dynamics, its basic idea is that the future evolution does not depend on the past history, but only on the present state. 
Following GKSL \cite{GoriniKossakowskiSudarshan1976,Lindblad1976}, we formulate this idea by the semigroup law for the time evolution maps: $\Lambda_{t+s}=\Lambda_t\circ\Lambda_s,\ t,s\geq 0$. Indeed, if the initial state is $\rho$ and the state at time $t$ is
$\rho_t:=\Lambda_t(\rho)$, then
\[
\rho_{t+\Delta t}
=
\Lambda_{t+\Delta t}(\rho)
=
(\Lambda_{\Delta t}\circ\Lambda_t)(\rho)
=
\Lambda_{\Delta t}(\rho_t).
\]

Thus the state at time $t+\Delta t$ is determined only by the state $\rho_t$ at time $t$. 

With this idea in mind, Kossakowski \cite{Kossakowski1972a,Kossakowski1972b} formulated quantum Markovian dynamics in terms of dynamical semigroups.

\begin{Def}
A one-parameter family of real-linear maps $\{\Lambda_t\}_{t\geq 0}$ on $\LA_{\rm sa}(\HA)$ is called a \emph{dynamical semigroup} if it satisfies the following conditions:
\begin{enumerate}
\item[(i)] (State preservation): $\Lambda_t$ is positive and trace-preserving for all $t\geq 0$;
\item[(ii)] (Semigroup law): $\Lambda_{t+s}=\Lambda_t\circ\Lambda_s$ for all $t,s\geq 0$;
\item[(iii)] (Strong continuity): $\lim_{t\downarrow 0}\|\Lambda_t(X)-X\|_1=0$ for all $X\in\LA_{\rm sa}(\HA)$.
\end{enumerate}
If, in addition, the complex-linear extension of each $\Lambda_t$ to $\LA(\HA)$ is completely positive, then $\{\Lambda_t\}_{t\geq 0}$ is called a \emph{CP dynamical semigroup} \cite{GoriniKossakowskiSudarshan1976}.
\end{Def}
In order to distinguish positivity from complete positivity clearly, we shall use the following terminology in this note. 
A trace-preserving positive semigroup will be called a \emph{positive dynamical semigroup}, whereas a trace-preserving completely positive semigroup will be called a \emph{completely positive (CP) dynamical semigroup}.
In both cases, trace preservation is included in the terminology.

\begin{Rem}
\leavevmode\par
\begin{enumerate}
\item A positive dynamical semigroup is thus a strongly continuous semigroup on the real
normed vector space $\LA_{\rm sa}(\HA)$.

\item The strong continuity condition in (iii) means continuity at the initial time, where one naturally sets $\Lambda_0=\I$. Moreover, together with the semigroup law, it implies the continuity condition \eqref{eq:TimeConti} at every time.

\item In the finite-dimensional setting considered in this note, the particular choice of norm is not essential. In fact, strong continuity is equivalent to uniform continuity,
\[
\lim_{t\downarrow 0}\|\Lambda_t-\I\|_{1 \to 1}=0.
\]
For this reason, although we state the continuity condition in the trace norm, we shall often use the equivalent uniform continuity formulation.
\end{enumerate}
See Sec.~\ref{LPThm} and \ref{app:semigroup} for details.
\end{Rem}
Since $\LA_{\rm sa}(\HA)$ is finite-dimensional, the continuity condition (iii),
together with the semigroup law, ensures the existence of a generator
$\mathfrak{L}$ defined on the whole space $\LA_{\rm sa}(\HA)$ by
\begin{equation}\label{eq:L}
    \mathfrak{L}(X)
    :=
    \lim_{t\downarrow 0}
    \frac{\Lambda_t(X)-X}{t},
    \qquad X\in\LA_{\rm sa}(\HA).
\end{equation}
The dynamics $\rho\mapsto\rho_t:=\Lambda_t(\rho)$ is then described by the
master equation
\begin{equation}
    \frac{d}{dt}\rho_t=\mathfrak{L}(\rho_t).
\end{equation}
Equivalently,
\begin{equation}
    \Lambda_t=\exp(t\mathfrak{L})
    :=\sum_{n=0}^{\infty}\frac{t^n\mathfrak{L}^n}{n!}.
\end{equation}
For completeness, we give a simple proof of the existence of the generator in \ref{app:Gen}.

\section{Derivation of the GKLS generator}

The following theorem is Kossakowski's 1972 characterization of generators of positive dynamical semigroups in the present sense, namely positive trace-preserving semigroups; see Theorem~5 in Ref.~\cite{Kossakowski1972b}, where the result is stated in a different but equivalent form.
\begin{Thm}[Kossakowski 1972]\label{thm:K}
A linear map $\mathfrak{L}$ on ${\cal L}_{\rm sa}(\HA)$ is the generator of a positive dynamical semigroup if and only if the following conditions hold:
\begin{subequations}\label{cond}
\begin{enumerate}
\item[(i)] (Positivity) For any mutually orthogonal projections $P$ and $Q$, i.e., $PQ=QP=0$,
\begin{equation}
\Tr [P(\mathfrak{L}(Q))] \geq 0.
\label{cond1}
\end{equation}

\item[(ii)] (Trace Preservation) For any projection $P$,
\begin{equation}
\Tr [\mathfrak{L}(P)] = 0.
\label{cond2}
\end{equation}
\end{enumerate}
\end{subequations}\end{Thm}
\begin{Rem}
At a formal level, the necessity of these conditions is quite transparent.
Indeed, expanding $\Lambda_t=e^{t\mathfrak{L}}$ for small $t$ gives $\Lambda_t=\id+t\mathfrak{L}+O(t^2)$. 
If $P$ and $Q$ are mutually orthogonal projections, then, since $PQ=0$,
\[
\Tr[P\Lambda_t(Q)]
=
t\Tr[P(\mathfrak{L}(Q))]+O(t^2).
\]
Positivity of $P$ and $\Lambda_t(Q)$ for all $t\geq 0$ implies that the left-hand side is nonnegative (use Fact~\ref{ex:TrAB} below). 
Dividing by $t>0$ and letting $t\downarrow 0$, it follows that
\[
\Tr[P(\mathfrak{L}(Q))]\geq 0.
\]
Similarly, trace preservation of $\Lambda_t$ implies $\Tr[\mathfrak{L}(P)]=0$ for every projection $P$.
The nontrivial part of Kossakowski's characterization is that these infinitesimal conditions are also sufficient for $e^{t\mathfrak{L}}$ to define a positive trace-preserving semigroup. 
\end{Rem}
\begin{F}\label{ex:TrAB}
If $A,B$ are self-adjoint operators, then $\Tr AB\in\R$.
If $A,B$ are positive operators, then $\Tr AB\geq 0$.
\end{F}
The celebrated representation theorem for the generator of a completely positive
dynamical semigroup was obtained by Gorini, Kossakowski, and Sudarshan as follows.
\begin{Thm}[Gorini--Kossakowski--Sudarshan 1976]\label{thm:GKS}
A linear map $\mathfrak{L}$ on $\LA(\HA)\simeq \LA(\CA^d)$ is the generator of a completely
positive dynamical semigroup on $\LA(\HA)$ if and only if it can be expressed in the form
\begin{align}\label{eq:GKSform}
\mathfrak{L}(\rho)
&=
-i[H,\rho]
+\frac{1}{2}\sum_{\alpha,\beta=1}^{d^2-1} c_{\alpha\beta}
\Bigl(
[F_\alpha,\rho F_\beta^\dagger]
+
[F_\alpha\rho,F_\beta^\dagger]
\Bigr)  \notag\\
&=
-i[H,\rho]
+\frac{1}{2}\sum_{\alpha,\beta=1}^{d^2-1} c_{\alpha\beta}
\Bigl(
2F_\alpha\rho F_\beta^\dagger
-
F_\beta^\dagger F_\alpha\rho
-
\rho F_\beta^\dagger F_\alpha
\Bigr),
\end{align}
where $H$ is a self-adjoint operator, called an effective Hamiltonian,
$\{F_\alpha\}_{\alpha=1}^{d^2-1}$ is an orthonormal basis of the space $\LA_0(\HA)$ of all traceless operators with respect to the Hilbert--Schmidt inner product, that is,
\[
\Tr F_\alpha=0,
\qquad
\Tr F_\alpha^\dagger F_\beta=\delta_{\alpha\beta}
\quad
(\alpha,\beta=1,\ldots,d^2-1),
\]
and $C=(c_{\alpha\beta})_{\alpha,\beta=1}^{d^2-1}$ is a positive matrix. 
The effective Hamiltonian $H$ is uniquely determined by the condition $\Tr H=0$.
\end{Thm}
By diagonalizing the positive matrix $(c_{\alpha\beta})$, one
obtains an orthonormal basis $\{G_\alpha\}_{\alpha=1}^{d^2-1}$ of
$\LA_0(\HA)$ and nonnegative numbers $\lambda_\alpha$
$(\alpha=1,\ldots,d^2-1)$ such that the expression \eqref{eq:GKSform} takes the diagonal form\footnote{\label{ft:Diag} Since $C=(c_{\alpha\beta})_{\alpha,\beta=1}^{d^2-1}$ is positive semidefinite, its eigenvalue decomposition gives a unitary matrix $U=(U_{\alpha\gamma})_{\alpha,\gamma=1}^{d^2-1}$ and nonnegative eigenvalues $\lambda_\gamma$ such that
\[
c_{\alpha\beta}
=
\sum_{\gamma=1}^{d^2-1}
\lambda_\gamma U_{\alpha\gamma}\overline{U_{\beta\gamma}}.
\]
Define
\[
G_\gamma
:=
\sum_{\alpha=1}^{d^2-1}
U_{\alpha\gamma}F_\alpha
\qquad
(\gamma=1,\ldots,d^2-1).
\]
By the unitarity of $U$,
\[
F_\alpha
=
\sum_{\gamma=1}^{d^2-1}
\overline{U_{\alpha\gamma}}G_\gamma.
\]

Substituting the diagonalization of $C$ into \eqref{eq:GKSform}, one obtains
\begin{align}
&\frac{1}{2}\sum_{\alpha,\beta=1}^{d^2-1} c_{\alpha\beta}
\Bigl(
[F_\alpha,\rho F_\beta^\dagger]
+
[F_\alpha\rho,F_\beta^\dagger]
\Bigr) =
\frac{1}{2}\sum_{\gamma=1}^{d^2-1}\lambda_\gamma
\sum_{\alpha,\beta=1}^{d^2-1}
U_{\alpha\gamma}\overline{U_{\beta\gamma}}
\Bigl(
[F_\alpha,\rho F_\beta^\dagger]
+
[F_\alpha\rho,F_\beta^\dagger]
\Bigr) \nonumber\\
&=
\frac{1}{2}\sum_{\gamma=1}^{d^2-1}\lambda_\gamma
\Bigl(
[G_\gamma,\rho G_\gamma^\dagger]
+
[G_\gamma\rho,G_\gamma^\dagger]
\Bigr).
\end{align}}
\begin{equation}\label{eq:GKSdia}
\mathfrak{L}(\rho)
=
-i[H,\rho]
+
\frac{1}{2}\sum_{\alpha=1}^{d^2-1}\lambda_\alpha
\Bigl(
[G_\alpha,\rho G_\alpha^\dagger]
+
[G_\alpha\rho,G_\alpha^\dagger]
\Bigr).
\end{equation}
Moreover, by setting $L_\alpha=\sqrt{\lambda_\alpha}G_\alpha$, one obtains
\begin{align}\label{LForm}
\mathfrak{L}(\rho)
&=
-i[H,\rho]
+
\frac{1}{2}\sum_{\alpha=1}^{d^2-1}
\Bigl(
2L_\alpha\rho L_\alpha^\dagger
-
L_\alpha^\dagger L_\alpha\rho
-
\rho L_\alpha^\dagger L_\alpha
\Bigr) \nonumber \\
&=-i[H,\rho]
+
\sum_{\alpha=1}^{d^2-1}
\Bigl(
L_\alpha\rho L_\alpha^\dagger
- \frac{1}{2}\{L_\alpha^\dagger L_\alpha,\rho\}
\Bigr).
\end{align}
This is essentially the form obtained independently by Lindblad \cite{Lindblad1976}. 
The operators $L_\alpha$ are often called noise operators or jump operators.

\subsection{Proof of Theorem \ref{thm:GKS}}

In this subsection, we prove Theorem~\ref{thm:GKS} by following the argument of \cite{GoriniKossakowskiSudarshan1976}.
The proof relies on Theorem~\ref{thm:K}, whose proof is postponed to the next section, and on the following lemmas. 
Throughout this section, we use the matrix units $E_{ij}=\ketbra{\phi_i}{\phi_j}$ associated with a fixed orthonormal basis $\{\phi_i\}_{i=1}^{d}$ of $\HA$.
For simplicity, we shall often write $E_{ij}=\ketbra{i}{j}$.
Note that $E_{ij}^\dagger=E_{ji}$ and $\sum_i E_{ii} = \sum_i \ketbra{i}{i} = \II$. 

\begin{Lem}\label{lem:cpDp}
$\{\Lambda_t\}_{t\geq 0}$ is a completely positive dynamical semigroup on
$\LA(\HA)$ if and only if $\{\Lambda_t\otimes\I_d\}_{t\geq 0}$ is a dynamical semigroup on $\LA(\HA)\otimes\LA(\HA)$.
\end{Lem}

\begin{proof}
By Theorem~\ref{thm:cp}-(ii), $\Lambda_t$ is completely positive if and only if
$\Lambda_t\otimes\I_d$ is positive. Thus, the positivity condition for $\Lambda_t\otimes\I_d$ is equivalent to the complete positivity of $\Lambda_t$.

We next check trace preservation. If $\Lambda_t$ is trace-preserving, then for any $X\in\LA(\HA)\otimes\LA(\HA)$ written as $X=\sum_i A_i\otimes B_i$, we have
\[
\Tr((\Lambda_t\otimes\I_d)(X))
=
\sum_i \Tr(\Lambda_t(A_i))\Tr B_i
=
\sum_i \Tr A_i \Tr B_i
=
\Tr X.
\]
Hence $\Lambda_t\otimes\I_d$ is trace-preserving. Conversely, if
$\Lambda_t\otimes\I_d$ is trace-preserving, then for any $Y\in\LA(\HA)$,
\[
\Tr(\Lambda_t(Y))\Tr\II
=
\Tr((\Lambda_t\otimes\I_d)(Y\otimes\II))
=
\Tr(Y\otimes\II)
=
\Tr Y \Tr\II.
\]
Since $\Tr\II=d\neq 0$, it follows that $\Tr(\Lambda_t(Y))=\Tr Y$. Thus
$\Lambda_t$ is trace-preserving.

The semigroup law is also equivalent. Indeed,
\[
(\Lambda_{t+s}\otimes\I_d)
=
(\Lambda_t\otimes\I_d)\circ(\Lambda_s\otimes\I_d)
\]
holds if and only if
\[
\Lambda_{t+s}=\Lambda_t\circ\Lambda_s,
\]
because the equality can be tested on elements of the form $Y\otimes\II$.

Finally, we check strong continuity. Suppose first that $\{\Lambda_t\}_{t\geq 0}$ is strongly continuous. 
Let $X\in\LA(\HA)\otimes\LA(\HA)$ and write $X=\sum_i A_i\otimes B_i.$
Then
\[
((\Lambda_t\otimes\I_d)(X)-X)
=
\sum_i(\Lambda_t(A_i)-A_i)\otimes B_i.
\]
Hence
\[
\|(\Lambda_t\otimes\I_d)(X)-X\|_1
\leq
\sum_i \|(\Lambda_t(A_i)-A_i)\otimes B_i\|_1.
\]
Using $\|C\otimes B_i\|_1=\|C\|_1\|B_i\|_1$ (see Fact \ref{exe:TrnTrn} below), we obtain
\[
\|(\Lambda_t\otimes\I_d)(X)-X\|_1
\leq
\sum_i \|\Lambda_t(A_i)-A_i\|_1\|B_i\|_1
\to 0
\]
as $t\downarrow 0$. Thus $\{\Lambda_t\otimes\I_d\}_{t\geq 0}$ is strongly
continuous.

Conversely, suppose that $\{\Lambda_t\otimes\I_d\}_{t\geq 0}$ is strongly
continuous. For any $Y\in\LA(\HA)$, we have
\[
(\Lambda_t\otimes\I_d)(Y\otimes\II)-Y\otimes\II
=
(\Lambda_t(Y)-Y)\otimes\II.
\]
Therefore,
\[
d\|\Lambda_t(Y)-Y\|_1
=
\|(\Lambda_t(Y)-Y)\otimes\II\|_1
=
\|(\Lambda_t\otimes\I_d)(Y\otimes\II)-Y\otimes\II\|_1
\to 0
\]
as $t\downarrow 0$. Hence $\{\Lambda_t\}_{t\geq 0}$ is strongly continuous.
Combining these observations proves the claim.
\end{proof}
\begin{F}\label{exe:TrnTrn}
For $A,B \in \LA(\HA)$, $|A\otimes B|=|A|\otimes |B|$.
In particular, taking the trace gives $\|A\otimes B\|_1=\|A\|_1\|B\|_1$.
\end{F}

\begin{Lem}\label{lem:idF}
For any orthonormal basis $\{F_\alpha\}_{\alpha=1,\ldots,d^2}$ of $\LA(\HA)$
with respect to the Hilbert--Schmidt inner product, we have
\begin{equation}\label{eq:idF}
    \sum_{\alpha=1}^{d^2} F_\alpha^\dagger A F_\alpha
    =
    \Tr(A)\II
    \qquad (A\in\LA(\HA)).
\end{equation}
\end{Lem}

\begin{proof}
First note that
\[
    \sum_{\alpha=1}^{d^2} F_\alpha^\dagger A F_\alpha
\]
is independent of the choice of the Hilbert--Schmidt orthonormal basis $\{F_\alpha\}$ of $\LA(\HA)$. 
Indeed, if $\{F_\alpha\}$ and $\{E_\alpha\}$ are two orthonormal bases of $\LA(\HA)$, then there exists a unitary matrix $U=(U_{\alpha\beta})\in M_{d^2}(\CA)$ such that $E_\alpha=\sum_\beta U_{\alpha\beta}F_\beta.$
Therefore,
\[
\begin{aligned}
    \sum_\alpha E_\alpha^\dagger A E_\alpha
    &=
    \sum_{\beta,\gamma}
    \left(\sum_\alpha \overline{U_{\alpha\beta}}U_{\alpha\gamma}\right)
    F_\beta^\dagger A F_\gamma =
    \sum_\beta F_\beta^\dagger A F_\beta.
\end{aligned}
\]
Thus we may choose the matrix units $E_{ij}=\ketbra{i}{j}$ as an orthonormal basis of $\LA(\HA)$. 
Then, by the definition of the trace \eqref{eq:tr} and the completeness relation $\sum_j \ketbra{j}{j}=\II$, 
\[
\begin{aligned}
    \sum_{i,j=1}^d E_{ij}^\dagger A E_{ij}=
    \sum_{i,j=1}^d |j\rangle\langle i|A| i\rangle\langle j| =
    \sum_{i=1}^d \langle i|A|i\rangle
    \sum_{j=1}^d |j\rangle\langle j| =
    \Tr(A)\II.
\end{aligned}
\]
This proves the claim.
\end{proof}
\begin{Lem}\label{lem:BigONB}
Let $\{F_\alpha\}_{\alpha=1,\ldots,d^2}$ be an orthonormal basis of
$\LA(\HA)$. For $\alpha,\beta=1,\ldots,d^2$, define a linear map
$\Gamma_{\alpha\beta}$ on $\LA(\HA)$ by
\[
    \Gamma_{\alpha\beta}(A):=F_\alpha A F_\beta^\dagger .
\]
Then $\{\Gamma_{\alpha\beta}\}_{\alpha,\beta=1,\ldots,d^2}$ is an
orthonormal basis of $\LA(\LA(\HA))$ with respect to the Hilbert--Schmidt
inner product \eqref{eq:HSIL}.
\end{Lem}

\begin{proof}
Let $\{G_\lambda\}_{\lambda=1,\ldots,d^2}$ be any Hilbert--Schmidt
orthonormal basis of $\LA(\HA)$. Then
\begin{align*}
&\braket{\Gamma_{\alpha\beta}}{\Gamma_{\mu\nu}}_{\rm HS} =
\sum_\lambda
\Tr\left[
    (\Gamma_{\alpha\beta}(G_\lambda))^\dagger
    \Gamma_{\mu\nu}(G_\lambda)
\right] =
\sum_\lambda
\Tr\left[
    (F_\alpha G_\lambda F_\beta^\dagger)^\dagger
    F_\mu G_\lambda F_\nu^\dagger
\right] \text{(by \eqref{eq:HSIL}) }\\ 
&=
\sum_\lambda
\Tr\left[
    F_\beta G_\lambda^\dagger F_\alpha^\dagger
    F_\mu G_\lambda F_\nu^\dagger
\right] =
\sum_\lambda
\Tr\left[
    F_\nu^\dagger F_\beta G_\lambda^\dagger
    F_\alpha^\dagger F_\mu G_\lambda
\right] \qquad \text{(by cyclicity of the trace)} \\
&=
\Tr\left[
    F_\nu^\dagger F_\beta
    \sum_\lambda
    G_\lambda^\dagger F_\alpha^\dagger F_\mu G_\lambda
\right] =
\Tr\left[
    F_\nu^\dagger F_\beta
    \Tr(F_\alpha^\dagger F_\mu)\II
\right] \qquad \text{(by Lemma~\ref{lem:idF})}\\
&=
\Tr(F_\alpha^\dagger F_\mu)
\Tr(F_\nu^\dagger F_\beta) =
\delta_{\alpha\mu}\delta_{\beta\nu}.
\end{align*}
Thus $\{\Gamma_{\alpha\beta}\}_{\alpha,\beta=1,\ldots,d^2}$ is an orthonormal set in $\LA(\LA(\HA))$. Since $\LA(\LA(\HA))$ has dimension $d^4$ and the set consists of $d^4$ elements, it is an orthonormal basis.
\end{proof}

\begin{Lem}\label{lem:ONBform}
Let $\{F_\alpha\}_{\alpha=1,\ldots,d^2}$ be an orthonormal basis of $\LA(\HA)$. 
Any linear map $\Gamma$ on $\LA(\HA)$ can be uniquely written in the form
\begin{equation}
\Gamma(A)=\sum_{\alpha,\beta=1}^{d^2} c_{\alpha\beta}
F_\alpha A F_\beta^\dagger .
\end{equation}
Moreover, $\Gamma$ is Hermiticity-preserving if and only if
$C=(c_{\alpha\beta})$ is a Hermitian matrix, that is,
$\overline{c_{\beta\alpha}}=c_{\alpha\beta}$.
\end{Lem}

\begin{proof}
The first assertion follows directly from Lemma~\ref{lem:BigONB}. For the second assertion, we first note the following elementary characterization: 
a linear map $\Gamma$ is Hermiticity-preserving, that is, $\Gamma(A)\in\LA_{\rm sa}(\HA)$ for all $A\in\LA_{\rm sa}(\HA)$, if and only if $\Gamma(A)^\dagger=\Gamma(A^\dagger)$ for all $A\in\LA(\HA)$.
Indeed, the identity immediately implies that $\Gamma$ maps self-adjoint operators to self-adjoint operators. Conversely, suppose that $\Gamma$ is Hermiticity-preserving. 
Using the Cartesian decomposition \eqref{CDec}, write $A=A_R+iA_I$ with $A_R,A_I\in\LA_{\rm sa}(\HA)$. Then
\[
\Gamma(A)^\dagger
=
\Gamma(A_R+iA_I)^\dagger
=
\Gamma(A_R)^\dagger-i\Gamma(A_I)^\dagger
=
\Gamma(A_R)-i\Gamma(A_I)
=
\Gamma(A^\dagger),
\]
as desired. 
Using this characterization and the representation above, we obtain
\[
\Gamma(A)^\dagger
=
\sum_{\alpha,\beta=1}^{d^2}
\overline{c_{\alpha\beta}}F_\beta A^\dagger F_\alpha^\dagger
=
\sum_{\alpha,\beta=1}^{d^2}
\overline{c_{\beta\alpha}}F_\alpha A^\dagger F_\beta^\dagger.
\]
On the other hand,
\[
\Gamma(A^\dagger)
=
\sum_{\alpha,\beta=1}^{d^2}
c_{\alpha\beta}F_\alpha A^\dagger F_\beta^\dagger.
\]
By the uniqueness of the expansion in the first assertion, these two expressions coincide for all $A$ if and only if $\overline{c_{\beta\alpha}}=c_{\alpha\beta}$ for all $\alpha,\beta$. This is equivalent to saying that $C=(c_{\alpha\beta})$ is Hermitian.
\end{proof}
 
\begin{Prop}\label{prop:GKSformHerm}
A linear map $\mathfrak{L}$ on $\LA(\HA)$ is the generator of a
trace-preserving and Hermiticity-preserving semigroup if and only if it can be expressed in the form
\begin{equation}\label{eq:GKSformHerm}
\mathfrak{L}(A)
=
-i[H,A]
+
\frac{1}{2}\sum_{\alpha,\beta=1}^{d^2-1}
c_{\alpha\beta}
\Bigl(
[F_\alpha,AF_\beta^\dagger]
+
[F_\alpha A,F_\beta^\dagger]
\Bigr),
\end{equation}
where $H=H^\dagger$, $\{F_\alpha\}_{\alpha=1}^{d^2}$ is an orthonormal basis of $\LA(\HA)$ such that $F_{d^2}=\II/\sqrt{d}$, and $C=(c_{\alpha\beta})_{\alpha,\beta=1}^{d^2-1}$ is a Hermitian matrix. With the condition $\Tr H=0$, the effective Hamiltonian $H$ is uniquely determined.
\end{Prop}

\begin{proof}
Let $\mathfrak{L}$ be the generator of a trace-preserving and Hermiticity-preserving semigroup. 
By Lemma~\ref{lem:ONBform}, we can write
\begin{equation}
\mathfrak{L}(A)
=
\frac{1}{d}c_{d^2d^2}A
+
\frac{1}{\sqrt{d}}\sum_{\alpha=1}^{d^2-1}
\bigl(
c_{\alpha d^2}F_\alpha A
+
c_{d^2\alpha}AF_\alpha^\dagger
\bigr)
+
\sum_{\alpha,\beta=1}^{d^2-1}
c_{\alpha\beta}F_\alpha A F_\beta^\dagger ,
\end{equation}
where the full matrix $(c_{\alpha\beta})_{\alpha,\beta=1}^{d^2}$ is Hermitian.
Put
\[
F:=\frac{1}{\sqrt{d}}\sum_{\alpha=1}^{d^2-1}c_{\alpha d^2}F_\alpha,
\]
and write its Cartesian decomposition as $F=F_R+iF_I$, where $F_R=F_R^\dagger$ and $F_I=F_I^\dagger$. Since the coefficient matrix is Hermitian, the preceding expression becomes
\[
\mathfrak{L}(A)
=
\frac{1}{d}c_{d^2d^2}A
+
FA+AF^\dagger
+
\sum_{\alpha,\beta=1}^{d^2-1}
c_{\alpha\beta}F_\alpha A F_\beta^\dagger .
\]
Using $FA+AF^\dagger=i[F_I,A]+\{F_R,A\}$, we obtain
\[
\mathfrak{L}(A)
=
\frac{1}{d}c_{d^2d^2}A
+
i[F_I,A]
+
\{F_R,A\}
+
\sum_{\alpha,\beta=1}^{d^2-1}
c_{\alpha\beta}F_\alpha A F_\beta^\dagger .
\]
Let $H:=-F_I$ and
\[
G:=F_R+\frac{1}{2d}c_{d^2d^2}\II .
\]
Then
\begin{equation}\label{Ltemp}
\mathfrak{L}(A)
=
-i[H,A]
+
\{G,A\}
+
\sum_{\alpha,\beta=1}^{d^2-1}
c_{\alpha\beta}F_\alpha A F_\beta^\dagger .
\end{equation}

We now use trace preservation of the semigroup. Since $\exp(t\mathfrak{L})$ is trace-preserving for all $t\geq 0$, differentiating
$\Tr[\exp(t\mathfrak{L})(A)]=\Tr A$ at $t=0$ gives $\Tr\mathfrak{L}(A)=0$ for all $A\in\LA(\HA)$\footnote{
Here the trace is interchanged with the derivative.
Indeed, if $A_t$ is differentiable in trace norm, then, using
$|\Tr X|\leq \|X\|_1$, we have
\[
\left|
\frac{\Tr A_{t+h}-\Tr A_t}{h}
-
\Tr\left(\frac{d}{dt}A_t\right)
\right|
\leq
\left\|
\frac{A_{t+h}-A_t}{h}
-
\frac{d}{dt}A_t
\right\|_1
\to 0
\]
as $h\to 0$.
}.
Thus, taking the trace of equation~\eqref{Ltemp}, 
\[
0
=
\Tr\left[
\left(
2G+\sum_{\alpha,\beta=1}^{d^2-1}
c_{\alpha\beta}F_\beta^\dagger F_\alpha
\right)A
\right]
\]
for all $A\in\LA(\HA)$. Here we have used the cyclicity of the trace: $\Tr [H,A]=\Tr(HA-AH)=0$ and $\Tr\{G,A\}=\Tr(GA+AG)=2\Tr(GA)$.
Hence
\[
G
=
-\frac{1}{2}
\sum_{\alpha,\beta=1}^{d^2-1}
c_{\alpha\beta}F_\beta^\dagger F_\alpha.
\]
(See Fact \ref{ex:TrAB=0} below). Substituting this into \eqref{Ltemp} gives \eqref{eq:GKSformHerm}.

Conversely, suppose that $\mathfrak{L}$ has the form \eqref{eq:GKSformHerm}. 
A direct computation shows that $\mathfrak{L}(A^\dagger)=\mathfrak{L}(A)^\dagger$ and $\Tr\mathfrak{L}(A)=0$ for all $A\in\LA(\HA)$. Hence $\mathfrak{L}$ is Hermiticity-preserving and trace-annihilating. It follows that $\exp(t\mathfrak{L})$ is Hermiticity-preserving and trace-preserving for all $t\geq 0$. 
Therefore, $\mathfrak{L}$ generates a trace-preserving and Hermiticity-preserving semigroup.
\end{proof}

\begin{F}\label{ex:TrAB=0}
If $\Tr(BA)=0$ for all $A \in \LA(\HA)$, then $B=0$.
\end{F}

\begin{Rem} Since $C =(c_{\alpha \beta})$ is Hermitian, \eqref{eq:GKSformHerm} can be diagonalized. This allows the generator $\mathfrak{L}$ in \eqref{eq:GKSformHerm} to be written in the diagonal form 
\begin{align}\label{eq:GKSdiaHerm}
\mathfrak{L}(\rho) &= -i[H,\rho] + \frac{1}{2} \sum_{\alpha=1}^{d^2-1} \lambda_{\alpha} \Bigl( [G_\alpha,\rho G_\alpha^\dagger]  + [G_\alpha \rho, G_\alpha^\dagger]\Bigr) \nonumber  \\
&= -i[H,\rho] + \frac{1}{2} \sum_{\alpha=1}^{d^2-1} \lambda_{\alpha} \Bigl( 2 G_\alpha\rho G_\alpha^\dagger - G_\alpha^\dagger G_\alpha \rho - \rho  G_\alpha^\dagger G_\alpha \Bigr).
\end{align}
See footnote~\ref{ft:Diag} for the same derivation.
Here $\lambda_\alpha$ is an eigenvalue of the Hermitian matrix $C$, and hence is real.
It should be noted, however, that some of the $\lambda_\alpha$ may be negative at this stage, since only Hermiticity and trace preservation of the semigroup have been assumed.
\end{Rem}

\begin{Lem}\label{lem:cfmop}
Let
$\{\hat{P}^{(\alpha)}\}_{\alpha=1,\ldots,n}$ be a family of operators on
$\LA(\HA)\otimes\LA(\HA)$ of the form
\begin{equation}
\hat{P}^{(\alpha)}
=
\sum_{i,j=1}^{d}P^{(\alpha)}_{ij}\otimes E_{ij},
\end{equation}
where $E_{ij} = \ketbra{i}{j}$ is a matrix unit. 
Then $\{\hat{P}^{(\alpha)}\}_{\alpha=1,\ldots,n}$ is a family of mutually orthogonal projections if and only if
\begin{equation}
\left(P^{(\alpha)}_{ij}\right)^\dagger
=
P^{(\alpha)}_{ji},
\qquad
\sum_{l=1}^d
P^{(\alpha)}_{il}P^{(\beta)}_{lj}
=
\delta_{\alpha\beta}P^{(\alpha)}_{ij},
\qquad
(i,j=1,\ldots,d,\ \alpha,\beta=1,\ldots,n).
\label{eq:cfmop1}
\end{equation}
\end{Lem}
{\it Proof}. Note that an element of the form $\hat{P}=\sum_{i,j=1}^{d} P_{ij}\otimes E_{ij}$ is a projection if and only if
\[
P_{ij}^\dagger=P_{ji}
\quad\text{and}\quad
\sum_{l=1}^d P_{il}P_{lj}=P_{ij}
\qquad
(i,j=1,\ldots,d).
\]
Indeed, using $(A\otimes B)^\dagger = A^\dagger \otimes B^\dagger$ and $E_{ij}^\dagger=E_{ji}$, we have
\[
\hat{P}^\dagger
=
\left(\sum_{i,j=1}^{d} P_{ij}\otimes E_{ij}\right)^\dagger
=
\sum_{i,j=1}^{d} P_{ij}^\dagger\otimes E_{ji}
=
\sum_{i,j=1}^{d} P_{ji}^\dagger\otimes E_{ij}.
\]
Since $\{E_{ij}\}$ forms an orthonormal basis of $\LA(\HA)$, the expansion of $\hat{P}$ with respect to this basis is unique.
Therefore, $\hat{P}$ is self-adjoint if and only if $P_{ij}^\dagger=P_{ji}$ for all $i,j$. (See Fact~\ref{fact:PartTen} below.)
Moreover,
\[
\hat{P}^2
=
\left(\sum_{i,l=1}^{d} P_{il}\otimes E_{il}\right)
\left(\sum_{m,j=1}^{d} P_{mj}\otimes E_{mj}\right)
=
\sum_{i,l,m,j=1}^{d} P_{il}P_{mj}\otimes E_{il}E_{mj}.
\]
Since $E_{il}E_{mj}=\delta_{lm}E_{ij}$, this becomes
\[
\hat{P}^2
=
\sum_{i,j=1}^{d}
\left(
\sum_{l=1}^{d}P_{il}P_{lj}
\right)
\otimes E_{ij}.
\]
Hence, $\hat{P}^2=\hat{P}$ if and only if $\sum_{l=1}^{d}P_{il}P_{lj}=P_{ij}$ for all $i,j$.
Thus $\hat{P}$ is a projection if and only if the two stated conditions hold.

Note also that two such projections $\hat{P}$ and $\hat{Q}=\sum_{i,j=1}^{d} Q_{ij}\otimes E_{ij}$ are orthogonal $\hat{P}\hat{Q}=0$ if and only if
\[
\sum_{l=1}^d P_{il}Q_{lj}=0
\qquad
(i,j=1,\ldots,d).
\]
Indeed, 
\[
\hat{P}\hat{Q}
=
\left(\sum_{i,l=1}^{d}P_{il}\otimes E_{il}\right)
\left(\sum_{m,j=1}^{d}Q_{mj}\otimes E_{mj}\right)
=
\sum_{i,l,m,j=1}^{d}P_{il}Q_{mj}\otimes E_{il}E_{mj}.
\]
Since $E_{il}E_{mj}=\delta_{lm}E_{ij}$, this becomes
\[
\hat{P}\hat{Q}
=
\sum_{i,j=1}^{d}
\left(
\sum_{l=1}^{d}P_{il}Q_{lj}
\right)
\otimes E_{ij}.
\]
Thus $\hat{P}\hat{Q}=0$ if and only if $\sum_{l=1}^{d}P_{il}Q_{lj}=0$ for all $i,j$.

Applying these observations to the family $\{\hat{P}^{(\alpha)}\}_{\alpha=1,\ldots,n}$ proves the assertion.
\hfill $\square$

\begin{F}\label{fact:PartTen}
Let $\HA_1,\HA_2$ be Hilbert spaces of dimensions $d_1$ and $d_2$, respectively.
Let $\{\phi_i\}_{i=1}^{d_2}$ be an orthonormal basis of $\HA_2$.
If $\psi_i,\psi'_i\in\HA_1$, $i=1,\ldots,d_2$, satisfy
\[
\sum_{i=1}^{d_2}\psi_i\otimes\phi_i
=
\sum_{i=1}^{d_2}\psi'_i\otimes\phi_i,
\]
then $\psi_i=\psi'_i$ for all $i=1,\ldots,d_2$.
\end{F}

\begin{Lem}\label{lem:FEF}
Let $\{F_\alpha\}_{\alpha=1,\ldots,n}$ be an orthonormal system in $\LA(\HA)$ with respect to the Hilbert--Schmidt inner product, i.e.,
$\Tr F_\alpha^\dagger F_\beta=\delta_{\alpha\beta}$.
Define
\begin{equation}
\hat{P}_\alpha
=
\sum_{i,j=1}^{d} P^{(\alpha)}_{ij}\otimes E_{ij},
\qquad
P^{(\alpha)}_{ij}
=
F_\alpha E_{ij}F_\alpha^\dagger.
\end{equation}
Then $\{\hat{P}_\alpha\}_{\alpha=1,\ldots,n}$ is a family of mutually orthogonal projections on $\LA(\HA)\otimes\LA(\HA)$.
\end{Lem}

\begin{proof}
By Lemma~\ref{lem:cfmop}, it suffices to show that
$P^{(\alpha)}_{ij}=F_\alpha E_{ij}F_\alpha^\dagger$ satisfies the conditions in \eqref{eq:cfmop1}.
First,
\begin{equation}
(P^{(\alpha)}_{ij})^\dagger
=
(F_\alpha E_{ij}F_\alpha^\dagger)^\dagger
=
F_\alpha E_{ji}F_\alpha^\dagger
=
P^{(\alpha)}_{ji}.
\end{equation}
Moreover,
\begin{align}
\sum_{l=1}^d P^{(\alpha)}_{il}P^{(\beta)}_{lj}
&=
\sum_{l=1}^d
(F_\alpha E_{il}F_\alpha^\dagger)
(F_\beta E_{lj}F_\beta^\dagger) =
\sum_{l=1}^d
(F_\alpha \ketbra{\phi_i}{\phi_l}F_\alpha^\dagger)
(F_\beta \ketbra{\phi_l}{\phi_j}F_\beta^\dagger)
\nonumber\\
&=
\Tr(F_\alpha^\dagger F_\beta)
F_\alpha \ketbra{\phi_i}{\phi_j}F_\beta^\dagger=
\delta_{\alpha\beta}P^{(\alpha)}_{ij}.
\end{align}
Therefore the conditions in \eqref{eq:cfmop1} hold, and Lemma~\ref{lem:cfmop} proves the assertion.
\end{proof}

\begin{Lem}\label{lem:pos-matrix-trace}
Let $C=(c_{\alpha\beta})_{\alpha,\beta=1}^n$ be a positive matrix.
Then, for arbitrary operators $X_1,\ldots,X_n \in \LA(\HA)$,
\begin{equation}
\sum_{\alpha,\beta=1}^n
c_{\alpha\beta}
\Tr\left(X_\alpha X_\beta^\dagger\right)
\geq 0.
\end{equation}
\end{Lem}
\begin{proof}
Using an eigenvalue decomposition of the positive matrix $C$,
\[
c_{\alpha\beta}
=
\sum_{r=1}^n
\lambda_r u^{(r)}_\alpha \overline{u^{(r)}_\beta},
\]

\begin{align}
\sum_{\alpha,\beta=1}^n
c_{\alpha\beta}
\Tr\left(X_\alpha X_\beta^\dagger\right)
&=
\sum_{r=1}^n
\lambda_r
\sum_{\alpha,\beta=1}^n
u^{(r)}_\alpha \overline{u^{(r)}_\beta}
\Tr\left(X_\alpha X_\beta^\dagger\right)
\nonumber\\
&=
\sum_{r=1}^n
\lambda_r
\Tr\left[
\left(
\sum_{\alpha=1}^n u^{(r)}_\alpha X_\alpha
\right)
\left(
\sum_{\beta=1}^n u^{(r)}_\beta X_\beta
\right)^\dagger
\right]\geq 0.
\end{align}
The last inequality follows from $\lambda_r\geq 0$ and
$\Tr(YY^\dagger)\geq 0$ for any operator $Y$.

Alternatively, define $Y=(Y_{\alpha\beta})$ by
$Y_{\alpha\beta}:=\Tr(X_\beta X_\alpha^\dagger)$.
Since $Y$ is a positive matrix, Fact~\ref{ex:TrAB} implies
\[
\sum_{\alpha,\beta=1}^n
c_{\alpha\beta}
\Tr\left(X_\alpha X_\beta^\dagger\right)
=
\Tr(CY)
\geq 0.
\]
\end{proof}

{\it Proof of Theorem \ref{thm:GKS}}.
[If part]: 
Let $\{\Lambda_t\}_{t \geq 0} := \exp(t\mathfrak{L})$ be the semigroup generated by a generator of the form \eqref{eq:GKSform}.
By Lemma~\ref{lem:cpDp}, it is enough to show that $\mathfrak{L} \otimes \I_d$ satisfies the conditions of Theorem~\ref{thm:K}.
Since $\mathfrak{L} \otimes \I_d$ satisfies the trace-preservation condition \eqref{cond2}, it remains to verify the positivity condition \eqref{cond1} for mutually orthogonal projections of the form
\[
\hat{P}_1 = \sum_{i,j=1}^{d} P^{(1)}_{ij} \otimes E_{ij},
\qquad
\hat{P}_2 = \sum_{i,j=1}^{d} P^{(2)}_{ij} \otimes E_{ij},
\]
subject to the conditions \eqref{eq:cfmop1}.
We now observe that
\begin{align}
&\Tr \Bigl[ \hat{P}_1 (\mathfrak{L}\otimes \I_d) \hat{P}_2 \Bigr] =
\Tr \Biggl[
\left(\sum_{i,j=1}^{d} P^{(1)}_{ij} \otimes E_{ij}\right)
\left(\sum_{k,l=1}^{d} \mathfrak{L}(P^{(2)}_{kl}) \otimes E_{kl}\right)
\Biggr] =
\sum_{i,j=1}^{d}
\Tr \bigl[P^{(1)}_{ij}\mathfrak{L}(P^{(2)}_{ji})\bigr] \nonumber\\
&=
-i \sum_{i,j=1}^{d}
\Tr \bigl(P^{(1)}_{ij}[H,P^{(2)}_{ji}]\bigr)
+
\frac{1}{2}
\sum_{i,j=1}^{d}
\sum_{\alpha,\beta=1}^{d^2-1}
c_{\alpha\beta}
\Tr \Bigl[
P^{(1)}_{ij}
\bigl(
2F_\alpha P^{(2)}_{ji}F_\beta^\dagger
-
F_\beta^\dagger F_\alpha P^{(2)}_{ji}
-
P^{(2)}_{ji}F_\beta^\dagger F_\alpha
\bigr)
\Bigr] \nonumber\\
&\overset{\heartsuit}{=}
\sum_{i,j=1}^{d}
\sum_{\alpha,\beta=1}^{d^2-1}
c_{\alpha\beta}
\Tr \bigl(
P^{(1)}_{ij}F_\alpha P^{(2)}_{ji}F_\beta^\dagger
\bigr) \nonumber\\
&\overset{\spadesuit}{=} 
\sum_{i,j,k,l=1}^{d}
\sum_{\alpha,\beta=1}^{d^2-1}
c_{\alpha\beta}
\Tr \bigl(
P^{(1)}_{ik}P^{(1)}_{kj}
F_\alpha
P^{(2)}_{jl}P^{(2)}_{li}
F_\beta^\dagger
\bigr) \nonumber\\
&=
\sum_{k,l=1}^{d}
\sum_{\alpha,\beta=1}^{d^2-1}
c_{\alpha\beta}
\Tr \Biggl[
\left(
\sum_{j=1}^{d}
P^{(1)}_{kj}F_\alpha P^{(2)}_{jl}
\right)
\left(
\sum_{i=1}^{d}
P^{(1)}_{ki}F_\beta P^{(2)}_{il}
\right)^\dagger
\Biggr]
\geq 0.
\label{eq:ineq}
\end{align}
Here, in the fourth equality $\heartsuit$, we used the orthogonality relation in \eqref{eq:cfmop1} and the cyclic property of the trace.
In the fifth equality $\spadesuit$, we used the idempotence relation in \eqref{eq:cfmop1}.
The last inequality follows from Lemma~\ref{lem:pos-matrix-trace}, since $(c_{\alpha\beta})$ is a positive matrix.

\bigskip 

[Only if part]: Assume that $\mathfrak{L}$ generates a completely positive dynamical semigroup.
Since $\mathfrak{L}$ also generates a trace-preserving and Hermiticity-preserving semigroup, we can use Proposition~\ref{prop:GKSformHerm} and write $\mathfrak{L}$ in the diagonal form \eqref{eq:GKSdiaHerm} with an orthonormal basis $G_\alpha$ $(\alpha=1,\ldots,d^2-1)$ and $G_{d^2}=\frac{1}{\sqrt{d}}\II$.
It remains to show that all $\lambda_\alpha$ are nonnegative.
Applying Lemma~\ref{lem:FEF} to $\{G_\alpha\}_{\alpha=1,\ldots,d^2}$, we obtain a complete family of mutually orthogonal projections on $\LA(\HA)\otimes\LA(\HA)$ of the form
\begin{equation}\label{eq:PGForm}
\hat{P}_\alpha
=
\sum_{i,j=1}^{d} P^{(\alpha)}_{ij}\otimes E_{ij},
\qquad
P^{(\alpha)}_{ij}
=
G_\alpha E_{ij}G_\alpha^\dagger.
\end{equation}

By Lemma~\ref{lem:cpDp} and Theorem~\ref{thm:K}, for any $\alpha=1,\ldots,d^2-1$, we have
\begin{equation}
\Tr\left[\hat{P}_\alpha(\mathfrak{L}\otimes \I_d)\hat{P}_{d^2}\right] \geq 0.
\end{equation}
On the other hand, the left-hand side can be computed explicitly.
Indeed, using \eqref{eq:PGForm} in the diagonal form \eqref{eq:GKSdiaHerm}, the same computation as in the if part, up to the equality marked $\heartsuit$, gives
\begin{equation}
\Tr\left[\hat{P}_\alpha(\mathfrak{L}\otimes \I_d)\hat{P}_{d^2}\right]
=
\sum_{i,j=1}^{d}
\sum_{\gamma=1}^{d^2-1}
\lambda_\gamma
\Tr \bigl(
P^{(\alpha)}_{ij}G_\gamma P^{(d^2)}_{ji}G_\gamma^\dagger
\bigr).
\end{equation}
Therefore, using
$P^{(\alpha)}_{ij}=G_\alpha E_{ij}G_\alpha^\dagger$ and
$P^{(d^2)}_{ji}=G_{d^2}E_{ji}G_{d^2}^\dagger=\frac{1}{d}E_{ji}$, we obtain
\begin{align}
\Tr\left[\hat{P}_\alpha(\mathfrak{L}\otimes \I_d)\hat{P}_{d^2}\right]
&=
\frac{1}{d}
\sum_{\gamma=1}^{d^2-1}
\lambda_\gamma
\sum_{i,j=1}^{d}
\Tr\left[
G_\alpha E_{ij}G_\alpha^\dagger
G_\gamma E_{ji}G_\gamma^\dagger
\right] \nonumber\\
&=
\frac{1}{d}
\sum_{\gamma=1}^{d^2-1}
\lambda_\gamma
\sum_{i,j=1}^{d}
\Tr\left[
G_\alpha \ketbra{i}{j}G_\alpha^\dagger
G_\gamma \ketbra{j}{i}G_\gamma^\dagger
\right] \nonumber\\
&=
\frac{1}{d}
\sum_{\gamma=1}^{d^2-1}
\lambda_\gamma
\left(
\sum_{j=1}^{d}
\langle j | G_\alpha^\dagger G_\gamma | j\rangle
\right)
\left(
\sum_{i=1}^{d}
\langle i | G_\gamma^\dagger G_\alpha | i\rangle
\right) \nonumber\\
&=
\frac{1}{d}
\sum_{\gamma=1}^{d^2-1}
\lambda_\gamma
\Tr(G_\alpha^\dagger G_\gamma)
\Tr(G_\gamma^\dagger G_\alpha)
=
\frac{1}{d}\lambda_\alpha.
\end{align}
Hence $\lambda_\alpha\geq 0$ for every $\alpha=1,\ldots,d^2-1$, completing the proof.
\hfill $\square$

\subsection{Kossakowski's characterization of the generators of positive dynamical semigroups}

In this subsection, we prove Theorem~\ref{thm:K}, namely Kossakowski's characterization of the generators of positive dynamical semigroups.
We first present some basic lemmas in Sec.~\ref{BLem}.
We then review contractive semigroups and the Lumer--Phillips theorem
\cite{LumerPhillips1961} in Sec.~\ref{LPThm}.
Next, we discuss Kossakowski's observation that positive dynamical semigroups are trace-preserving contractive semigroups, as stated in Theorem~\ref{thm:PosCont}.
We then prove Theorem~\ref{thm:K} in Sec.~\ref{KosThm}.
Finally, in Sec.~\ref{sec:Kpr}, we add a few remarks on the Kossakowski product.

\subsubsection{Basic Lemmas}\label{BLem}

In this section, we shall frequently use the basic facts summarized in Sec.~\ref{sec:MP}, together with a few additional facts introduced below\footnote{If the logic of a proof is unclear, the reader is encouraged to refer back to that section as needed. }.
\begin{F}\label{ex:NormEigChar}
For a self-adjoint operator $A \in \LA_{\rm sa}(\HA)$, 
\begin{equation}
\Tr A = \sum_{i=1}^d a_i, \ 
\|A\|_1 = \sum_{i=1}^d |a_i|, \ 
\|A\|_{\rm HS} = \sqrt{\sum_{i=1}^d a_i^2}, \ \|A\|_\infty = \max_{i=1,\ldots, d} |a_i|
\end{equation}
where $a_i \ (i=1,\ldots,d)$ are eigenvalues of $A$. 
\end{F}

\begin{Lem}\label{lem:otnorm}
For any $A,B\in\LA_{\rm sa}(\HA)$,
\[
|\Tr AB|\leq \|A\|_\infty \|B\|_1.
\]
\end{Lem}

\begin{proof}
This inequality holds for arbitrary operators, but the proof is particularly simple for self-adjoint operators. 
Let $B=\sum_i b_i\ketbra{b_i}{b_i}$ be an eigenvalue decomposition of $B$. 
Then
\[
|\Tr AB|
=
\left|\sum_i b_i \langle b_i | Ab_i\rangle\right|
\leq
\sum_i |b_i||\langle b_i | Ab_i\rangle|
\leq
\sum_i |b_i|\|Ab_i\|
\leq
\|A\|_\infty \sum_i |b_i|
=
\|A\|_\infty \|B\|_1.
\]
Here we used the Cauchy--Schwarz inequality and the definition of the operator norm. 
\end{proof}

\begin{Lem}\label{lem:ptr}
For $\rho\in\LA_{\rm sa}(\HA)$, the following statements hold:
\begin{enumerate}
\item[(i)] $\rho\geq 0$ if and only if $\|\rho\|_1=\Tr\rho$.
\item[(ii)] $|\Tr\rho|\leq \Tr|\rho|=\|\rho\|_1$.
\end{enumerate}
\end{Lem}

\begin{proof}
Let $r_i$, $i=1,\ldots,d$, be the eigenvalues of the self-adjoint operator
$\rho$.

(i) The condition $\rho\geq 0$ is equivalent to $r_i\geq 0$ for all $i$. If
$\rho\geq 0$, then $\|\rho\|_1=\sum_i |r_i|=\sum_i r_i=\Tr\rho$. Conversely, if
$\|\rho\|_1=\Tr\rho$, then $\sum_i(|r_i|-r_i)=0$. Since each $|r_i|-r_i$ is
nonnegative, we have $|r_i|=r_i$ for all $i$, and hence $r_i\geq 0$ for all
$i$. Thus $\rho\geq 0$.

(ii) Since $\Tr\rho=\sum_i r_i$ and $\Tr|\rho|=\sum_i |r_i|$, the statement
follows from the triangle inequality
$|\sum_i r_i|\leq \sum_i |r_i|$.
\end{proof}

For a self-adjoint operator $\rho\in\LA_{\rm sa}(\HA)$ with eigenvalue
decomposition $\rho=\sum_i r_i\ketbra{r_i}{r_i}$, the positive part $\rho_+$ and the negative part $\rho_-$ are defined by
\[
\rho_+:=\sum_{i:r_i>0} r_i\ketbra{r_i}{r_i},
\qquad
\rho_-:=\sum_{i:r_i<0} |r_i|\ketbra{r_i}{r_i},
\]
respectively. Then $\rho_+$ and $\rho_-$ are positive and orthogonal to each other, that is, $\rho_\pm\geq 0$ and $\rho_+\rho_-=\rho_-\rho_+=0$. From the definitions, we immediately obtain the following decomposition.
\begin{Lem}\label{lem:posneg}
Any self-adjoint $\rho\in\LA_{\rm sa}(\HA)$ has a positive-negative decomposition
\begin{equation}
\rho = \rho_+ - \rho_-,
\label{decPM}
\end{equation}
where $\rho_+,\rho_-\geq 0$ and $\rho_+\rho_-=0$. Moreover,
\begin{equation}
\|\rho\|_1 = \|\rho_+\|_1 + \|\rho_-\|_1.
\label{decPMTNorm}
\end{equation}
\end{Lem}

\bigskip

\subsubsection{Contractive semigroups and the Lumer--Phillips Theorem}\label{LPThm}

In this subsection, we temporarily let $V$ denote an arbitrary finite-dimensional normed space over $\mathbb{K}=\R$ or $\CA$, equipped with a norm $\|\cdot\|$. 
This convention is used only here and in \ref{app:semigroup}. For our purposes, it is enough to keep in mind the case $V=\LA_{\rm sa}(\HA)$ equipped with the trace norm $\|\cdot\|_1$.

\begin{Def}\label{def:CntSemi}
A one-parameter family of linear maps $\{\Lambda_t\}_{t\geq 0}$ on $V$ is called a \emph{contractive semigroup} if it satisfies the following conditions:
\begin{enumerate}
\item[(i)] \emph{Contractivity}: $\|\Lambda_t(v)\|\leq \|v\|$ for all
$v\in V$ and all $t\geq 0$;
\item[(ii)] \emph{Semigroup law}: $\Lambda_{t+s}=\Lambda_t\circ\Lambda_s$ for all $t,s\geq 0$;
\item[(iii)] \emph{Strong continuity}:
$\lim_{t\downarrow 0}\|\Lambda_t(v)-v\|=0$ for all $v\in V$.
\end{enumerate}
\end{Def}

\begin{Def}\label{def:SIP}
An L-semi-inner product on $V$ \cite{Lumer1961} is a binary operation $V\times V\to\mathbb{K}$, denoted by $\<v,v'\>$, satisfying the following conditions:\begin{subequations}\label{eq:si}
\begin{align}
\<v,\alpha v'+\beta v''\>
&=
\alpha\<v,v'\>+\beta\<v,v''\>, \label{eq:sia}\\
\<v,v\>& = \|v\|^2 \label{eq:sib}\\
|\<v,v'\>|&\leq
\|v\|\|v'\|, \label{eq:sic}
\end{align}
\end{subequations}
for all $v,v',v''\in V$ and $\alpha,\beta\in\mathbb{K}$.
\end{Def}
\begin{Rem}
An L-semi-inner product may be viewed as a generalization of an inner product. 
In particular, the inequality \eqref{eq:sic} is the analogue of the
Cauchy--Schwarz inequality. In Lumer's original paper \cite{Lumer1961}, it is simply called a semi-inner product. 
However, the term ``semi-inner product'' is also often used for an inner product without the positive-definiteness condition $\langle v,v\rangle=0 \Rightarrow v=0$. To avoid possible confusion, we call the present object an L-semi-inner product. 

More precisely, in Lumer's original paper, a semi-inner product is first introduced on a vector space, with the positivity condition $\langle v,v\rangle>0$ for $v\neq 0$ in place of \eqref{eq:sib}. 
The norm is then defined by $\|v\|:=\sqrt{\langle v,v\rangle}$, and the resulting space becomes a normed space. Moreover, Lumer showed that every normed space admits an L-semi-inner product compatible with the given norm. Such an L-semi-inner product is not necessarily unique.

In \cite{Kossakowski1972b}, by contrast, Kossakowski started with a normed space $(V,\|\cdot\|)$ and introduced a semi-scalar product compatible with the given norm in the sense of \eqref{eq:sib}.
\end{Rem}

\begin{Def}\label{def:dissipative}
Let $V$ be a normed space equipped with an L-semi-inner product. 
A linear map $\mathfrak{L}$ on $V$ is called dissipative if
\begin{equation}
\Re\<v,\mathfrak{L}v\>\leq 0
\quad
\text{for all } v\in V.
\end{equation}
\end{Def}

The Lumer--Phillips theorem \cite{LumerPhillips1961} is a celebrated result in semigroup theory, characterizing the generators of contractive semigroups in terms of dissipativity. In the finite-dimensional setting needed here, it takes the following simple form.

\begin{Thm}[Lumer--Phillips theorem, finite-dimensional form]\label{thm:LP}
Let $V$ be a finite-dimensional normed space equipped with an L-semi-inner product. 
A linear map $\mathfrak{L}:V\to V$ generates a contractive semigroup
$\{e^{t\mathfrak{L}}\}_{t\geq 0}$ if and only if $\mathfrak{L}$ is dissipative, that is,
\[
\operatorname{Re}\<v,\mathfrak{L}v\>\leq 0,
\qquad v\in V.
\]
\end{Thm}

For completeness, the proof is provided in \ref{app:LP}.

\subsubsection{Proof of Theorem \ref{thm:K}}\label{KosThm} 

One of Kossakowski's key ideas is to connect the contraction property with the positivity-preserving property.
\begin{Thm}\label{thm:PosCont}
A trace-preserving linear map $\Lambda$ on $\LA_{\rm sa}(\HA)$ is positive if and only if it is contractive with respect to the trace norm.
\end{Thm}
\begin{proof}
First suppose that $\Lambda$ is positive. 
For any $\rho\in\LA_{\rm sa}(\HA)$, write $\rho=\rho_+-\rho_-$ with $\rho_+,\rho_-\geq 0$, as in \eqref{decPM}.
Since $\Lambda$ is positive, both $\Lambda(\rho_+)$ and $\Lambda(\rho_-)$ are positive. 
Hence, we have
\begin{align*}
\|\Lambda(\rho)\|_1
&=
\|\Lambda(\rho_+)-\Lambda(\rho_-)\|_1 \quad \qquad \text{(by \eqref{decPM} and linearity of $\Lambda$)} \\
&\leq 
\|\Lambda(\rho_+)\|_1+\|\Lambda(\rho_-)\|_1 \ \quad \text{(by triangle inequality)} \\
&=
\Tr\Lambda(\rho_+)+\Tr\Lambda(\rho_-) \qquad \text{(by Lemma \ref{lem:ptr} (i) and $\Lambda(\rho_\pm) \ge 0$) } \\
&=
\Tr\rho_+ + \Tr\rho_- \qquad\qquad\quad \text{(by trace preserving property)} \\
&=
\|\rho_+\|_1+\|\rho_-\|_1 =
\|\rho\|_1 \quad \text{(by Lemma \ref{lem:ptr} (i) and \eqref{decPMTNorm}). }
\end{align*}
Thus $\Lambda$ is contractive.

Conversely, suppose that $\Lambda$ is contractive. We show that $\Lambda$ is positive. 
Let $\rho\geq 0$. By Lemma~\ref{lem:ptr}, the trace-preserving property of $\Lambda$, and contractivity, we have
\begin{align*}
\|\rho\|_1 = \Tr\rho = \Tr\Lambda(\rho) \leq \|\Lambda(\rho)\|_1 \leq \|\rho\|_1.
\end{align*}
Therefore all inequalities are equalities, and in particular
$\|\Lambda(\rho)\|_1=\Tr\Lambda(\rho)$. By Lemma~\ref{lem:ptr} (i), this implies
$\Lambda(\rho)\geq 0$. Hence $\Lambda$ is positive.
\end{proof}
\bigskip 

Thus, a positive dynamical semigroup is equivalently a trace-preserving contractive semigroup.
\begin{Cor}\label{cor:DyCont}
A one-parameter family $\{\Lambda_t\}_{t\geq 0}$ on $\LA_{\rm sa}(\HA)$ is a positive dynamical semigroup if and only if it is a trace-preserving contractive semigroup with respect to the trace norm.
\end{Cor}

To apply the Lumer--Phillips theorem, Kossakowski introduced the following L-semi-inner product on $\LA_{\rm sa}(\HA)$ with respect to the trace norm.
\begin{Def}
For $A,B\in \LA_{\rm sa}(\HA)$, the \emph{Kossakowski product} is defined by
\begin{equation}\label{KosProd}
\<A,B\>_K:=\|A\|_1 \Tr\left[\operatorname{sgn}(A)B\right],
\end{equation}
where the sign operator $\operatorname{sgn}(A)\in\LA_{\rm sa}(\HA)$ is defined by
\begin{equation}\label{eq:sgn}
\operatorname{sgn}(A):= P_+-P_-,
\end{equation}
where $P_+$ and $P_-$ denote the projections onto the positive and negative eigenspaces of $A$, respectively. Equivalently, if $A=\sum_\mu \lambda_\mu P_\mu$ is the spectral decomposition of $A$, $\operatorname{sgn}(A)=\sum_\mu \operatorname{sgn}(\lambda_\mu)P_\mu$, where the scalar sign function $\operatorname{sgn}(x)$ is given by
\[
\operatorname{sgn}(x)
:=
\begin{cases}
1 & x>0,\\
0 & x=0,\\
-1 & x<0.
\end{cases}
\]
\end{Def}

\begin{Prop}\label{prop:KossakowskiProduct}
The Kossakowski product is a real L-semi-inner product on $\LA_{\rm sa}(\HA)$ equipped with the trace norm.
\end{Prop}
\begin{proof} Since \( \operatorname{sgn}(A) \) is self-adjoint and the trace of a product of self-adjoint operators is real, we have \( \< A, B \>_K\in \mathbb{R} \). 
Let \( P_+ \), \( P_- \), and \( P_0 \) denote the spectral projectors of \( A \) corresponding to the positive, negative, and zero eigenvalues, respectively. By completeness of the spectral decomposition, we have \( P_+ + P_0 + P_- = \II \). 
We now verify that the Kossakowski product satisfies the three properties of a semi-inner product in Definition~\ref{def:SIP}:

First, linearity~\eqref{eq:sia} follows immediately from the linearity of the trace. 
Second, normalization~\eqref{eq:sib} is shown by observing 
\[
\Tr\left[ \operatorname{sgn}(A)\, A \right] 
= \Tr\left[(P_+ - P_-) A \right] 
= \sum_{i: a_i > 0} a_i - \sum_{i: a_i < 0} a_i 
= \sum_{i} |a_i| = \|A\|_1, 
\]
where $a_1,\ldots,a_d$ are eigenvalues of $A$.  Hence
\[
\<A,A\>_K
=
\|A\|_1\Tr\left[\operatorname{sgn}(A)A\right]
=
\|A\|_1^2.
\]
Finally, we prove the Cauchy--Schwarz-type inequality~\eqref{eq:sic}. 
Using the triangle inequality and an eigenvalue decomposition \( B = \sum_j b_j \ketbra{b_j}{b_j} \), we observe
\begin{align*}
|\Tr[(P_+ - P_-) B]| &\le |\Tr[P_+ B]| + |\Tr[P_- B]| \\
&= \left| \sum_j b_j \bra{b_j} P_+ \ket{b_j} \right| 
   + \left| \sum_j b_j \bra{b_j} P_- \ket{b_j} \right|\\
&\le \sum_j |b_j| \left( \bra{b_j} P_+ \ket{b_j} + \bra{b_j} P_- \ket{b_j} + \bra{b_j} P_0 \ket{b_j} \right) \\
&= \sum_j |b_j| \bra{b_j} (P_+ + P_- + P_0) \ket{b_j} \\
&= \sum_j |b_j| \bra{b_j} \II \ket{b_j} = \sum_j |b_j| = \|B\|_1.
\end{align*}

Therefore, we get 
\[
|\< A, B \>_K| = \|A\|_1 \left| \Tr\left[ (P_+ - P_-)\, B \right] \right| \le \|A\|_1 \|B\|_1,
\]
as required.
\end{proof}

We are now ready to prove Theorem~\ref{thm:K} by applying the Lumer--Phillips theorem with the Kossakowski product. Before entering the proof, however, we need one more simple observation concerning the passage to the limit.

The trace is a continuous linear functional on the finite-dimensional space $\LA_{\rm sa}(\HA)$. Hence the trace may be interchanged with limits. 
Namely, if $A_t\to A$ in trace norm as $t\downarrow 0$, then
\[
\lim_{t\downarrow 0}\Tr A_t=\Tr A
=
\Tr\left(\lim_{t\downarrow 0}A_t\right).
\]
(Here $t\downarrow 0$ reflects the semigroup setting $t\geq 0$; the same
argument applies to any trace-norm limit $A_t\to A$.)
In particular, if $A_t$ is differentiable in trace norm, then\footnote{Indeed, using $|\Tr A|\leq \|A\|_1$ for $A\in\LA_{\rm sa}(\HA)$ (see Lemma \ref{lem:ptr}), we have
\[
\left|
\frac{\Tr A_{t+h}-\Tr A_t}{h}
-
\Tr\left(\frac{d}{dt}A_t\right)
\right|
=
\left|
\Tr\left(
\frac{A_{t+h}-A_t}{h}
-
\frac{d}{dt}A_t
\right)
\right|
\leq
\left\|
\frac{A_{t+h}-A_t}{h}
-
\frac{d}{dt}A_t
\right\|_1
\to 0
\]
as $h\to 0$.}
\[
\frac{d}{dt}\Tr A_t
=
\Tr\left(\frac{d}{dt}A_t\right).
\]

\bigskip 

{\it Proof of Theorem~\ref{thm:K}.} We first show the necessity of (i) and (ii) in Theorem~\ref{thm:K}. 
Suppose that $\mathfrak{L}$ is a generator of a dynamical (i.e., trace preserving and positive) semigroup $\{\Lambda_t\}$. 

Since the limit defining $\mathfrak{L}$ in \eqref{eq:L} exists (as shown in \ref{app:semigroup}), and since the trace is continuous,
we have, for any $A\in\LA_{\rm sa}(\HA)$,
\[
\Tr\mathfrak{L}(A)
=
\lim_{t\downarrow 0}
\Tr\left[
\frac{\Lambda_t(A)-A}{t}
\right]
=
\lim_{t\downarrow 0}
\frac{\Tr\Lambda_t(A)-\Tr A}{t}
=
0.
\]
Here the last equality follows from the trace-preserving property. 
Since this argument applies to every self-adjoint $A\in\LA_{\rm sa}(\HA)$, it applies in particular to every projection $P$. 
Therefore condition (ii) in Theorem~\ref{thm:K} is satisfied.

Next, let $P$ and $Q$ be orthogonal projections. Then, we have
\begin{equation}\label{eq:TrPLQ}
\Tr\left[P\mathfrak{L}(Q)\right]
=
\lim_{t\downarrow 0}
\Tr\left[
P\frac{\Lambda_t(Q)-Q}{t}
\right]
=
\lim_{t\downarrow 0}
\frac{1}{t}
\Tr\left[P\Lambda_t(Q)\right],
\end{equation}
where we have used $PQ=0$. 
Since $\Lambda_t$ is positive and $Q\geq 0$, we have $\Lambda_t(Q)\geq 0$ and $\Tr\left[P\Lambda_t(Q)\right]\geq 0$ for every $t>0$ (see Fact \ref{ex:TrAB}). Thus, the left-hand side of \eqref{eq:TrPLQ} is also non-negative. 
Therefore condition (i) in Theorem~\ref{thm:K} is satisfied.

Conversely, suppose that conditions (i) and (ii) in Theorem~\ref{thm:K} hold.
We show that the semigroup generated by $\mathfrak{L}$ is a positive dynamical semigroup.

First, we note that \eqref{cond2}, assumed for every projection $P$, extends by linearity to every self-adjoint operator $A\in\LA_{\rm sa}(\HA)$. 
Indeed, if $A=\sum_\mu \lambda_\mu P_\mu$ is the spectral decomposition of $A$, then
\[
\Tr\mathfrak{L}(A)
=
\sum_\mu \lambda_\mu \Tr\mathfrak{L}(P_\mu)
=
0.
\]
This immediately implies that $\Lambda_t:=e^{t\mathfrak{L}}$ is trace
preserving. 
Indeed, for $A\in\LA_{\rm sa}(\HA)$, since $\frac{d}{dt}\Lambda_t(A)=\mathfrak{L}(\Lambda_t(A))$ and $\Lambda_t(A)\in\LA_{\rm sa}(\HA)$, and since, as observed above, the trace commutes with differentiation, we have
\[
\frac{d}{dt}\Tr\Lambda_t(A)
=
\Tr\mathfrak{L}(\Lambda_t(A))
=
0.
\]
Hence $\Tr\Lambda_t(A)$ is constant in $t$. Since $\Lambda_0=\id$, we obtain $\Tr\Lambda_t(A)=\Tr A$ for all $t\geq 0$. Thus $\Lambda_t$ is trace preserving.

Next, we verify that $\mathfrak{L}$ is dissipative with respect to the Kossakowski product:
\begin{equation}\label{cond:dis}
\< A,\mathfrak{L}(A) \>_K
=
\Re \< A,\mathfrak{L}(A) \>_K
\leq 0
\qquad
(\forall A\in \LA_{\rm sa}(\HA)).
\end{equation}
To show this, let $A=\sum_\mu \lambda_\mu P_\mu$ be the spectral decomposition of $A$, where the $P_\mu$ are mutually orthogonal projections with $\sum_\mu P_\mu=\II$. 
Then $\operatorname{sgn}(A)=\sum_\mu \operatorname{sgn}(\lambda_\mu)P_\mu$, and hence
\begin{align}
\< A,\mathfrak{L}(A) \>_K
&=
\|A\|_1
\Tr\left[
\left(\sum_\mu \operatorname{sgn}(\lambda_\mu)P_\mu\right)
\mathfrak{L}\left(\sum_\nu \lambda_\nu P_\nu\right)
\right]
\nonumber\\
&=
\|A\|_1
\sum_\mu \operatorname{sgn}(\lambda_\mu)\lambda_\mu
\Tr[P_\mu\mathfrak{L}(P_\mu)] + \|A\|_1
\sum_\mu\sum_{\nu\neq\mu}
\operatorname{sgn}(\lambda_\nu)\lambda_\mu
\Tr[P_\nu\mathfrak{L}(P_\mu)].
\label{eq:ala}
\end{align}
Since $\sum_\nu P_\nu=\II$, condition (ii) in Theorem~\ref{thm:K} gives
\[
0
=
\Tr[\mathfrak{L}(P_\mu)]
=
\sum_\nu\Tr[P_\nu\mathfrak{L}(P_\mu)]
=
\Tr[P_\mu\mathfrak{L}(P_\mu)]
+
\sum_{\nu\neq\mu}\Tr[P_\nu\mathfrak{L}(P_\mu)].
\]
Hence
\[
\Tr[P_\mu\mathfrak{L}(P_\mu)]
=
-\sum_{\nu\neq\mu}\Tr[P_\nu\mathfrak{L}(P_\mu)].
\]
Substituting this into \eqref{eq:ala}, we obtain
\begin{equation}\label{ALAii}
\< A,\mathfrak{L}(A) \>_K
=
-\|A\|_1
\sum_\mu\sum_{\nu\neq\mu}
\alpha_{\mu\nu}
\Tr[P_\nu\mathfrak{L}(P_\mu)],
\end{equation}
where
\[
\alpha_{\mu\nu}
:=
\left(
\operatorname{sgn}(\lambda_\mu)
-
\operatorname{sgn}(\lambda_\nu)
\right)\lambda_\mu.
\]
Since condition (i) in Theorem~\ref{thm:K} guarantees
$\Tr[P_\nu\mathfrak{L}(P_\mu)]\geq 0$ in \eqref{ALAii}, it remains only to show that $\alpha_{\mu\nu}\geq 0$. 
This is immediate if $\lambda_\mu=0$, since then $\alpha_{\mu\nu}=0$.
Suppose therefore that $\lambda_\mu\neq 0$.
Then $\lambda_\mu=\operatorname{sgn}(\lambda_\mu)|\lambda_\mu|$ and
$\operatorname{sgn}(\lambda_\mu)^2=1$.
Thus
\[
\begin{aligned}
\alpha_{\mu\nu}
&=
\left(
\operatorname{sgn}(\lambda_\mu)
-
\operatorname{sgn}(\lambda_\nu)
\right)
\lambda_\mu \\
&=
\left(
\operatorname{sgn}(\lambda_\mu)
-
\operatorname{sgn}(\lambda_\nu)
\right)
\operatorname{sgn}(\lambda_\mu)|\lambda_\mu| \\
&=
\left(
1-\operatorname{sgn}(\lambda_\mu)\operatorname{sgn}(\lambda_\nu)
\right)
|\lambda_\mu|. 
\end{aligned}
\]
Since $\operatorname{sgn}(x)\in\{-1,0,1\}$, the factor $1-\operatorname{sgn}(\lambda_\mu)\operatorname{sgn}(\lambda_\nu)$ is non-negative. 
Thus $\alpha_{\mu\nu}\geq 0$, and hence
$\<A,\mathfrak{L}(A)\>_K\leq 0$. This proves the dissipative condition \eqref{cond:dis}. 

Therefore, since $\mathfrak{L}$ is dissipative with respect to the Kossakowski product, the Lumer--Phillips theorem, Theorem~\ref{thm:LP}, implies that $\Lambda_t=e^{t\mathfrak{L}}$ is a contraction semigroup. By Corollary~\ref{cor:DyCont}, $\Lambda_t$ is therefore a positive dynamical semigroup. 
This completes the proof. \hfill $\square$ 

\begin{Rem}
One might wonder whether the usual Hilbert--Schmidt inner product could be used instead of the Kossakowski product.
This is not the case. 
Let us illustrate this point already in a purely classical two-dimensional example.
Consider the generator
\begin{equation}
\mathfrak{L}
\begin{pmatrix}
a & 0\\
0 & b
\end{pmatrix}
=
\begin{pmatrix}
b & 0\\
0 & -b
\end{pmatrix}, 
\end{equation}
or equivalently
\[
\frac{d}{dt}
\begin{pmatrix}
a(t)\\
b(t)
\end{pmatrix}
=
\begin{pmatrix}
b(t)\\
-b(t)
\end{pmatrix}.
\]

Solving this differential equation, we obtain
\begin{equation}
\Lambda_t
\begin{pmatrix}
a & 0\\
0 & b
\end{pmatrix}
=
\begin{pmatrix}
a+(1-e^{-t})b & 0\\
0 & e^{-t}b
\end{pmatrix}.
\end{equation}
This is a positive trace-preserving semigroup on diagonal Hermitian matrices, since it preserves $a+b$ and maps $a,b\geq 0$ to nonnegative diagonal entries.
Hence, by Theorem \ref{thm:PosCont}, it is contractive with respect to the trace norm.

Now take, for example, a positive matrix $A=\begin{pmatrix}2&0\\0&1\end{pmatrix}$.
The Hilbert--Schmidt inner product gives
\begin{equation}
\operatorname{Tr}[A\mathfrak{L}(A)]
=
\operatorname{Tr}
\left[
\begin{pmatrix}
2 & 0\\
0 & 1
\end{pmatrix}
\begin{pmatrix}
1 & 0\\
0 & -1
\end{pmatrix}
\right]
=
1>0.
\end{equation}
Thus the Hilbert--Schmidt dissipativity condition would fail, even though the semigroup is trace-norm contractive.
\end{Rem}

\subsection{A comment on the Kossakowski product}\label{sec:Kpr}

We end this section with a brief comment on the Kossakowski product \eqref{KosProd}.
The following discussion is meant to provide an intuitive and heuristic way of understanding the origin of the Kossakowski product.
It is partly interpretative, and therefore should not be read as a historical claim about Kossakowski's actual line of thought.
Nevertheless, in view of the role played by the trace norm and the Lumer--Phillips theorem, it is natural to think that this underlying idea was at least implicitly present in Kossakowski's thinking.

Suppose that $\Lambda_t=e^{t\mathfrak{L}}=\I+t\mathfrak{L}+O(t^2)$ is a contraction semigroup on the real Banach space $\mathcal{L}_{\mathrm{sa}}(\mathcal{H})$ equipped with the trace norm.
Then, for each $A\in\mathcal{L}_{\mathrm{sa}}(\mathcal{H})$, one expects the trace norm not to increase infinitesimally:
\begin{equation}\label{ContCond}
\left.
\frac{d}{dt}
\|A+t\mathfrak{L}(A)\|_1
\right|_{t=0+}
\leq 0
\end{equation}
for all $A\in\mathcal{L}_{\mathrm{sa}}(\mathcal{H})$.

The key observation is that, for $A\in\mathcal{L}_{\mathrm{sa}}(\mathcal{H})$ with no zero eigenvalues and for all $B\in\mathcal{L}_{\mathrm{sa}}(\mathcal{H})$, one has
\begin{equation}\label{ContKPr}
\left.
\frac{d}{dt}
\|A+tB\|_1
\right|_{t=0}
=
\operatorname{Tr}[\operatorname{sgn}(A)B].
\end{equation}
Multiplication by the normalization factor $\|A\|_1$ gives precisely the structure captured by the Kossakowski product \eqref{KosProd}.

To show \eqref{ContKPr}, let $A=\sum_i a_i \ketbra{\phi_i}{\phi_i}$ be the eigenvalue decomposition of $A$, where $a_i\neq 0$ for all $i$.
For simplicity, we first assume that $A$ has no degenerate eigenvalues.
For sufficiently small $t$, the eigenvalues of $A+tB$ remain separated from zero.
Let $a_i(t)$ denote the eigenvalue of $A+tB$ satisfying $a_i(0)=a_i$.

Note that, if $x(t)$ is differentiable at $t=0$ and $x(0)\neq 0$, then
\begin{equation}
\left.
\frac{d}{dt}
|x(t)|
\right|_{t=0}
=
\operatorname{sgn}(x(0))
\left.
\frac{d}{dt}
x(t)
\right|_{t=0},
\end{equation}
because the sign of $x(t)$ is constant for sufficiently small $t$.
By first-order perturbation theory, we have\footnote{
Indeed, let $\ket{\phi_i(t)}$ be a normalized eigenvector of $A+tB$ corresponding to $a_i(t)$, chosen so that $\ket{\phi_i(0)}=\ket{\phi_i}$.
Then
\begin{equation}
(A+tB)\ket{\phi_i(t)}=a_i(t)\ket{\phi_i(t)}.
\end{equation}
Differentiating both sides at $t=0$, we obtain
\begin{equation}
B\ket{\phi_i}
+
A
\left.
\frac{d}{dt}\ket{\phi_i(t)}
\right|_{t=0}
=
\left.
\frac{d}{dt}a_i(t)
\right|_{t=0}
\ket{\phi_i}
+
a_i
\left.
\frac{d}{dt}\ket{\phi_i(t)}
\right|_{t=0}.
\end{equation}
Multiplying this equation from the left by $\bra{\phi_i}$ and using $\bra{\phi_i}A=a_i\bra{\phi_i}$, the terms containing $\left.d\ket{\phi_i(t)}/dt\right|_{t=0}$ cancel and one gets \eqref{PEig}.}
\begin{equation}\label{PEig}
\left.
\frac{d}{dt}
a_i(t)
\right|_{t=0}
=
\bra{\phi_i}B\ket{\phi_i}
=
\operatorname{Tr}\left[\ketbra{\phi_i}{\phi_i}B\right].
\end{equation}
Since $A+tB$ is Hermitian, its trace norm is the sum of the absolute values of its eigenvalues.
Therefore,
\begin{align}
\left.
\frac{d}{dt}
\|A+tB\|_1
\right|_{t=0}&=
\left.
\frac{d}{dt}
\sum_i |a_i(t)|
\right|_{t=0} =
\sum_i \operatorname{sgn}(a_i)
\left.
\frac{d}{dt}
a_i(t)
\right|_{t=0} \nonumber\\
&=
\sum_i \operatorname{sgn}(a_i)
\operatorname{Tr}\left[\ketbra{\phi_i}{\phi_i}B\right] =
\operatorname{Tr}[\operatorname{sgn}(A)B].
\end{align}
The degenerate case can be treated similarly.
Indeed, within each degenerate eigenspace, the first-order corrections are given by the eigenvalues of the restriction of $B$ to that eigenspace, and their sum gives the trace of this restriction.
Therefore the same formula remains valid.

Taking $B=\mathfrak{L}(A)$ in \eqref{ContKPr}, the infinitesimal contraction condition \eqref{ContCond} becomes
\begin{equation}
\operatorname{Tr}[\operatorname{sgn}(A)\mathfrak{L}(A)]\leq 0.
\end{equation}
Multiplying both sides by $\|A\|_1$, this is equivalent to
\begin{equation}
\< A,\mathfrak{L}(A)\>_K \leq 0.
\end{equation}
Thus the Kossakowski product packages the differential form of trace-norm contractivity into the algebraic dissipativity condition appearing in the Lumer--Phillips theorem.

From this viewpoint, the Kossakowski product is not an ad hoc device.
Rather, it is a concrete realization, for the trace norm on Hermitian operators, of the supporting functional that appears in the Lumer--Phillips theory.
It converts the infinitesimal contractivity of a trace-norm contraction semigroup into an algebraic condition on its generator.

When $A$ has zero eigenvalues, however, the supporting functional of the trace norm at $A$ is no longer uniquely determined.
This non-uniqueness indeed reflects the general non-uniqueness of L-semi-inner products compatible with a given norm \cite{Lumer1961,LumerPhillips1961}.
In the trace-norm case, the ambiguity appears in the choice of the action of the supporting functional on the kernel of $A$.
Kossakowski's choice corresponds to taking $\operatorname{sgn}(A)$ to be zero on $\ker A$.
This choice is arguably the most natural one, since it is canonical, treats the positive and negative spectral subspaces symmetrically, and depends only on the spectral decomposition of $A$.

\section{Personal recollections}

From 2003 to 2004, I stayed at Nicolaus Copernicus University in Toru\'n, Poland, under the supervision of Professor Kossakowski\footnote{Strictly speaking, from October 2003 to March 2004, I stayed there as a doctoral student, and from April to September 2004, as a JSPS postdoctoral fellow under a commissioned research arrangement. I would like to express my sincere gratitude to Professor Hiromichi Nakazato of Waseda University, who was my supervisor in Japan at that time, and to the Japan Society for the Promotion of Science for their support.}. As far as I remember, I was the only postdoc in his laboratory at that time, and I had the extremely precious opportunity to discuss physics and mathematics with him almost every day.

Looking back, I now realize more deeply how special that situation was. I was still only a young researcher. Nevertheless, I was allowed to learn directly from one of the founders of the theory of open quantum systems.

The first thing I would like to convey about Professor Kossakowski is his warm and gentle personality. He was almost always smiling, and with his slightly playful manner he naturally created a bright and relaxed atmosphere around him.

I still remember vividly the day I first arrived in Toru\'n. At that time, I had never travelled abroad for research, and I went to Poland alone. He had kindly arranged to meet me at the bus stop. However, the bus from Warsaw to Toru\'n was delayed. I did not have a mobile phone, and so I had no way of contacting him. In fact, I was not even fully sure whether I was on the right bus.

I think the bus was more than two hours late. When I finally arrived at the bus stop, however, Professor Kossakowski was still there waiting for me. He greeted me with a smile and gave me a firm but warm hug.

There was something almost mysterious about Professor Kossakowski's warmth. It was something I had never quite felt in anyone else. It is difficult to explain this in words. Yet this warmth was evident not only in my own interactions with him, but also in the way he related to other researchers.

Relations among researchers can take many forms. Researchers may be competitors, colleagues, friends or sometimes even comrades. 
Precisely because they take their research seriously, however, there can also be a certain sharpness, a kind of tension, in their interactions. 
Even when a disagreement is purely intellectual, a passionate scientific debate may sometimes make the atmosphere tense.

With Professor Kossakowski, however, this never seemed to happen. Those who spoke with him were naturally drawn into his warm smile, and before long they too would begin to smile. 
There was something remarkable in his presence, a mysterious power that seemed quietly to embrace those around him and make them feel cheerful and at ease. 
Indeed, I often saw, in the faces of those who spoke with him, a softening of their caution, and a gentleness and respect different from their usual expression.

I would also like to say a little about his way of thinking and his attitude toward research. 
On weekdays, Professor Kossakowski would usually come to the room where I was staying at around ten or eleven in the morning. 
We would then discuss physics and mathematics, and I would report to him what I had been thinking about at that time. Before I knew it, this had become our daily routine.

Almost every morning, he would ask me with a smile, ``Any news?" To be honest, for me at that time, this was a kind of pressure, though in a positive sense. 
But afterwards he would always say, again with a smile, ``Are you sleeping well and eating properly? I have a responsibility to make Gen survive."

Professor Kossakowski would almost always bring with him some calculations he had done himself. In his usual modest way, he would say something like, ``I tried a little calculation." 
Many of the calculations he showed me looked very simple and elementary, using only $2\times 2$ or $3 \times 3$ matrices. 
At first, I thought that he was deliberately choosing simple examples out of consideration for me, since I did not understand difficult mathematics very well.

\begin{figure}[ht]
\centering
\includegraphics[width=10cm]{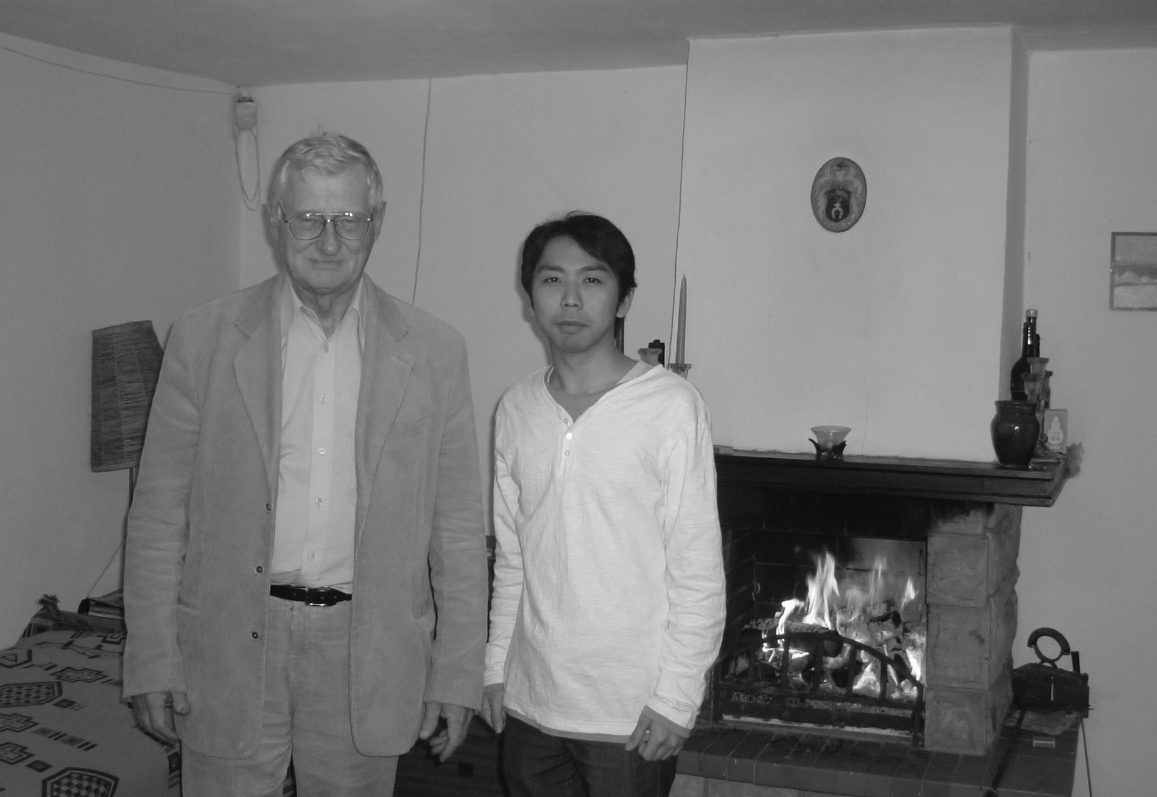}
\caption{With Professor Andrzej Kossakowski in Toru\'n, September 2004.}
\label{fig:KosKim}
\end{figure}
One day, however, Professor Kossakowski and his wife, Aniela, kindly invited me to their cottage and treated me to sushi they had made themselves.
This remains one of my most precious memories from my stay in Toru\'n; see Fig.~\ref{fig:KosKim}. On that occasion, he also showed me his study. 
There I saw a large number of sheets of calculation paper, many of them filled with calculations involving, again, small matrices.

When I saw this, I realized that he was not necessarily showing me elementary examples simply in order to adapt to my level. 
Rather, I think he himself regarded the thorough investigation of small and simple examples as an essential starting point for research.

As a small digression, I also had the chance to meet his beloved dog, Sara, on that occasion. I still remember, with great affection, how he gently murmured, ``She understands everything. She just does not tell me." 
I myself am also a dog lover and have rescue dogs. Even now, when I face a difficult problem, I sometimes imitate him and ask my dogs, Tsubaki, Yuzu and Rin. 
I must admit, however, that so far they have not given me an answer either.

Returning to the main point, at that time I was struggling to understand quantum mechanics in a physically intuitive way. 
To be honest, this is something I still struggle with even now. 
It was for this reason that I began to study mathematics on my own, hoping at least to understand the mathematical structure of the theory.
For someone like me, Professor Kossakowski, who made major contributions to the rigorous mathematical foundations of open quantum systems, including the infinite-dimensional case, by using highly advanced mathematics, was almost a heroic figure. 
For this reason, I was deeply moved to learn that behind such a grand and sophisticated theory there lay an accumulation of such concrete and elementary calculations.

Thinking about it now, I realize that the core structure of a theory is often condensed precisely in simple and beautiful examples. At that time, I do not think I understood this sufficiently. 
What Professor Kossakowski was looking at through his calculations with small matrices was not merely a collection of special cases. Rather, it was an essential structure leading toward a general theory. 
This also seems deeply connected with the fact that Gorini, Kossakowski, and Sudarshan were able, within a finite-dimensional framework, to grasp the essence of the generators of quantum Markovian dynamics.

Advanced theories are not born in an abstract form from the very beginning. Rather, they are born from calculating simple examples again and again, and from patiently observing the structures hidden within them. 
This was one of the most important lessons I learned from Professor Kossakowski.

At the same time, I also learned that the finite-dimensional world itself contains many nontrivial and profound problems. This is particularly evident in the theory of positive maps. 
On one occasion, Professor Kossakowski told me about the theory of positive maps: ``$2 \times 2$ matrices are easy, but as soon as you go to $3 \times 3$ matrices, things suddenly become nontrivial!" 
The fact that, even in finite dimensions, an entirely different world can appear simply by changing the dimension continues to fascinate me even now.

On another day, when we were talking about the relation between mathematics and physics, he said with a smile, half jokingly, ``Mathematicians live in a birdcage and never try to get out." 
I was surprised, because, as I mentioned above, in my limited understanding at the time, he had seemed to me almost like a mathematician. 
But thinking back on it now, I believe what he wanted to convey was that one should not remain confined within mathematics for its own sake. 
Rather, one must keep one's eyes firmly on the problems of physics, and when necessary, have the courage to move forward boldly, even beyond the safe region of a fully established mathematical framework. 
I think this is what he was trying to teach me.

This episode seems to show how deeply he valued physics. At the same time, it also shows that he was a researcher who elevated physical insight into something solid through logic and mathematics, while still having the courage to move forward when necessary.

There is still so much more I could say about Professor Kossakowski, but there is no end to it, so I will stop here for the moment. 
Professor Kossakowski passed away on February 1, 2021\footnote{Wikipedia \cite{WikipediaKossakowski} gives January 31 as the date of his passing, but according to what I was told by his wife and daughter, the actual date was February 1. 
Here I follow the date given by his family.}. I feel a deep sadness that my beloved mentor is no longer with us, that I can no longer see his smile in person, and that I can no longer receive his strong, warm hug. 
Even now, however, I feel a certain warmth, as if he were still nearby, quietly helping me with my research. 

\section{Concluding remarks}

In this note, I have tried to present the route to the GKLS equation followed by Kossakowski and by Gorini, Kossakowski, and Sudarshan in as logically self-contained a manner as possible.
The main point was not merely to reproduce the final form of the generator, but to make visible the structural logic behind it: positivity and trace preservation first lead to Kossakowski's infinitesimal characterization, and the additional requirement of complete positivity then singles out the positive Kossakowski matrix in the GKS representation.
In this sense, the GKLS equation is not just a convenient formula for open quantum dynamics, but the endpoint of a remarkably clear line of thought connecting operational requirements, semigroup theory, and the algebraic structure of finite-dimensional quantum mechanics.

I have also tried to emphasize the role of the Kossakowski product in this logical route.
Although it may look at first like a technical device introduced for the proof, it expresses in a concrete form the supporting functional of the trace norm and thereby connects positivity-preserving dynamics with contractive semigroups.
This observation, at least in the finite-dimensional setting discussed here, helps clarify why the Lumer--Phillips viewpoint fits so naturally with Kossakowski's characterization of positive dynamical semigroups.

The personal recollections included in the latter part of this note are not separate from this mathematical story.
What I learned from Professor Kossakowski was not only a theorem or a formalism, but a way of approaching physics: to look carefully at simple examples, to search for the structure hidden behind them, and to turn that structure into a clear mathematical statement.
The GKLS equation itself may be regarded as a beautiful realization of this attitude.

From this perspective, Kossakowski's legacy lies not only in the formula that bears his name, but also in a way of thinking that brought physics and mathematics together in a sound and fruitful balance.
It also lies in the example he set, not only as a scientist, but as a person of warmth, integrity, and deep humanity.
If this note succeeds even modestly in conveying both the mathematical clarity and the human warmth of that legacy, it will have achieved its purpose.

\bigskip

\section*{Acknowledgements}
I would like to express my deepest gratitude to Professor Kossakowski for the generosity with which he accepted me and for the warmth with which he guided me. His guidance has remained with me to this day, and I believe it will continue to accompany me in the years ahead. I offer my heartfelt prayers for his soul. 
I am deeply grateful to Professor Dariusz Chru\'sci\'nski, both as a collaborator and as a dear friend in Poland. 
I also thank Professor Paolo Muratore-Ginanneschi and Professor Saverio Pascazio for carefully reading the manuscript and for giving me valuable advice.
Finally, I would like to express my sincere gratitude to Professor Kossakowski's wife, Aniela Kossakowska, and his daughter, Aleksandra (Ola) Kossakowska, for allowing me to include some private memories in this note. Above all, I am profoundly grateful for the warmth with which they treat me like a member of their family, and for the support they have given me over the years.

\appendix
\renewcommand{\thesection}{Appendix \Alph{section}}

\section{Characterizations of Complete Positivity}\label{app:cp}

[Proof of Theorem \ref{thm:cp}]

\begin{proof}
By definition, (i) implies (ii).
If $\Lambda$ is $d$-positive, then $\Lambda\otimes\id_d$ maps positive operators on $\HA\otimes\CA^d$ to positive operators.
Since $\ketbra{\Psi_M}{\Psi_M}\geq 0$, this gives (iii).

It remains to show that (iii) implies (i).
Let 
\[
C_\Lambda:=d(\Lambda\otimes\id_d)(\ketbra{\Psi_M}{\Psi_M})=\sum_{i,j=1}^d \Lambda(E_{ij})\otimes E_{ij}
\]
be the Choi matrix of $\Lambda$.
By assumption, $C_\Lambda\geq 0$.
Therefore, by the eigenvalue decomposition, there exist vectors $v_\alpha\in\HA\otimes\CA^d$ such that
\[
C_\Lambda
=
\sum_\alpha
\ketbra{v_\alpha}{v_\alpha}.
\]
For each $\alpha$, write
\[
v_\alpha
=
\sum_{i=1}^d
f_{\alpha i}\otimes e_i
\qquad
(f_{\alpha i}\in\HA).
\]
Define an operator $M_\alpha\in\LA(\HA)$ by
\[
M_\alpha e_i=f_{\alpha i}
\qquad
(i=1,\ldots,d).
\]
Then
\[
\ketbra{v_\alpha}{v_\alpha}
=
\sum_{i,j=1}^d
\ketbra{f_{\alpha i}}{f_{\alpha j}}\otimes E_{ij}
=
\sum_{i,j=1}^d
M_\alpha E_{ij}M_\alpha^\dagger\otimes E_{ij}.
\]
Hence
\[
C_\Lambda
=
\sum_{i,j=1}^d
\left(
\sum_\alpha M_\alpha E_{ij}M_\alpha^\dagger
\right)
\otimes E_{ij}.
\]
Since $\{E_{ij}\}_{i,j=1}^d$ is a basis of $\LA(\CA^d)$, the expansion with respect to this basis is unique.
Therefore,
\[
\Lambda(E_{ij})
=
\sum_\alpha M_\alpha E_{ij}M_\alpha^\dagger
\qquad
(i,j=1,\ldots,d).
\]
By linearity, for every $A\in\LA(\HA)$,
\begin{equation}\label{eq:KrausRepPr}
\Lambda(A)
=
\sum_\alpha M_\alpha A M_\alpha^\dagger.
\end{equation}
This representation immediately implies complete positivity. 
Indeed, for any $n\in\mathbb{N}$ and any positive operator $X\in\LA(\HA\otimes\CA^n)$, we have
\[
(\Lambda\otimes\id_n)(X)
=
\sum_\alpha
(M_\alpha\otimes\II_n)X(M_\alpha^\dagger\otimes\II_n)
\geq 0.
\]
(See Fact \ref{XposYXY} below). 
Thus $\Lambda$ is completely positive.
This proves (iii) implies (i), and hence the three conditions are equivalent.
\end{proof}
\begin{F}\label{XposYXY}
If $X$ is a positive operator, then $YXY^\dagger$ is also positive for any operator $Y$.
\end{F}

\newpage 
\noindent{[Proof of Theorem~\ref{thm:cp2}]}

\begin{proof}
In the above proof of Theorem~\ref{thm:cp}, we have already shown that complete positivity of $\Lambda$ is equivalent to the existence of a Kraus representation \eqref{eq:KrausRepPr}. 
It remains to prove the trace-preserving condition.

Suppose that $\Lambda$ admits the Kraus representation \eqref{eq:KrausRepPr}. 
Then, for every $A\in\LA(\HA)$, using the linearity and cyclic property of the trace, we obtain
\[
\Tr\Lambda(A)
=
\sum_k \Tr(M_k A M_k^\dagger)
=
\sum_k \Tr(M_k^\dagger M_k A)
=
\Tr\left[\left(\sum_k M_k^\dagger M_k\right)A\right].
\]
Therefore, $\Lambda$ is trace-preserving, namely $\Tr\Lambda(A)=\Tr A$ for all $A\in\LA(\HA)$, if and only if
\[
\Tr\left[\left(\sum_k M_k^\dagger M_k\right)A\right]
=
\Tr(\II A)
\qquad
(A\in\LA(\HA)).
\]
By Fact~\ref{f:TrAB} below, this is equivalent to $\sum_k M_k^\dagger M_k=\II.$
\end{proof}
\begin{F}\label{f:TrAB}
If $\Tr XA = \Tr XB$ for all $X \in \LA(\HA)$, then $A = B$. 
\end{F}

\section{Semigroups in finite-dimensional normed spaces}\label{app:semigroup}

In Kossakowski's original papers, the characterization of positive dynamical semigroups, including the infinite-dimensional case, was formulated within the framework of the Hille--Yosida theorem and the Lumer--Phillips theorem. Since the present note is meant to remain within finite dimensions, we give in this appendix simple proofs tailored to the finite-dimensional setting: one for the existence of the generator of a semigroup, and the other for the Lumer--Phillips theorem.

In what follows, let $V$ be a finite-dimensional normed vector space over $\mathbb{K} =\R, \CA$, equipped with a norm $\|\cdot\|$.
For the applications in the main text, it is enough to take $V=\LA_{\rm sa}(\HA)$, the real vector space of self-adjoint operators on $\HA$, equipped with the trace norm $\|\cdot\|_1$. 
For a linear map $\Lambda$ on $V$, we define its operator norm by
\begin{equation}
\|\Lambda\|_{\mathrm{op}}
:=
\sup_{\substack{v\in V\\ v\neq 0}}
\frac{\|\Lambda v\|}{\|v\|}.
\end{equation}
\begin{F}\label{ex:SMul}
Let $V$ be a normed vector space.
For linear maps $\Lambda$ and $\Gamma$ on $V$, the operator norm is submultiplicative:
\[
\|\Lambda\Gamma\|_{\rm op}
\leq
\|\Lambda\|_{\rm op}\|\Gamma\|_{\rm op}.
\]
Moreover, if $\I$ denotes the identity map on $V$, then
\[
\|\I\|_{\rm op}=1.
\]
\end{F}

\begin{Def}
A family $\{\Lambda_t\}_{t\geq 0}$ of linear operators on $V$ is called a
one-parameter semigroup if it satisfies the following conditions:
\begin{enumerate}
\item {\it Initial condition}: $\Lambda_0=\I$;
\item {\it Semigroup law}: $\Lambda_{t+s}=\Lambda_t\Lambda_s$ for all $t,s\geq 0$.
\end{enumerate}
Moreover, it is called a strongly continuous one-parameter semigroup if it
further satisfies
\begin{enumerate}
\setcounter{enumi}{2}
\item {\it Strong continuity}:
\[
    \lim_{t\downarrow 0}\|\Lambda_t v-v\|=0
    \qquad \text{for all } v\in V.
\]
\end{enumerate}
\end{Def}
Note that, under the semigroup law, condition 3, strong continuity, already
implies condition 1, namely $\Lambda_0=\I$.
Since $V$ is finite-dimensional, strong continuity and uniform continuity are
equivalent. Thus the third condition may equivalently be replaced by
\begin{enumerate}
\setcounter{enumi}{3}
\item {\it Uniform continuity}:
\[
\lim_{t\downarrow 0}\|\Lambda_t-\I\|_{\rm op}=0.
\]
\end{enumerate}

The right-continuity condition in 3 or 4 is stated only at $t=0$.
However, this is enough to imply continuity of the semigroup at every
$t_0>0$, both from the right and from the left. Here we briefly verify this
under the uniform continuity condition 4. (The proof under the strong
continuity condition 3 is essentially the same; only the notation becomes
slightly more cumbersome. In any case, as noted above, in finite dimensions
all norms induce the same topology, so this distinction is immaterial for the
present discussion.)

First, by the uniform continuity 4, there exists $\delta>0$ such that $\|\Lambda_t-\I\|_{\rm op}\leq 1$ for
$0\leq t\leq \delta$. Hence $\|\Lambda_t\|_{\rm op}\leq \|\Lambda_t-\I\|_{\rm op}+\|\I\|_{\rm op}\leq 2$ for $0\leq t\leq \delta$. By the semigroup law, this implies that $\|\Lambda_t\|_{\rm op}$ is bounded on every compact interval $[0,T]$.
Indeed, for any $t\in[0,T]$, write $t=k\delta+r$ with an integer $k\geq 0$ and $0\leq r<\delta$. Then
\[
\Lambda_t=(\Lambda_\delta)^k\Lambda_r,
\qquad
\|\Lambda_t\|_{\rm op}
\leq
\|\Lambda_\delta\|^k_{\rm op}\|\Lambda_r\|_{\rm op}
\leq
2^{k+1}
\leq
2^{\lceil T/\delta\rceil+1},
\qquad
0\leq t\leq T.
\]
Now fix $t_0>0$. For $h>0$, the semigroup law gives $\Lambda_{t_0+h}-\Lambda_{t_0}= \Lambda_{t_0}(\Lambda_h-\id)$, and hence
\[
    \|\Lambda_{t_0+h}-\Lambda_{t_0}\|_{\rm op}
    \leq
    \|\Lambda_{t_0}\|_{\rm op}\|\Lambda_h-\id\|_{\rm op}
    \to 0
    \qquad
    (h\downarrow 0).
\]
Thus $\Lambda_t$ is right-continuous at $t_0$.
Similarly, for $0<h<t_0$, we have $\Lambda_{t_0}=\Lambda_{t_0-h}\Lambda_h$, 
so that
\[
    \Lambda_{t_0-h}-\Lambda_{t_0}
    =
    \Lambda_{t_0-h}(\id-\Lambda_h).
\]
Since $\|\Lambda_{t_0-h}\|_{\rm op}$ is bounded for $h$ sufficiently small, it follows that
\[
    \|\Lambda_{t_0-h}-\Lambda_{t_0}\|_{\rm op}
    \leq
    \|\Lambda_{t_0-h}\|_{\rm op}\|\Lambda_h-\id\|_{\rm op}
    \to 0
    \qquad
    (h\downarrow 0).
\]
Hence $\Lambda_t$ is also left-continuous at $t_0$. Therefore the continuity
condition at $t=0$ implies right-continuity at $t=0$ and continuity at every
$t>0$.

\subsection{Existence of the generator of a semigroup}\label{app:Gen}

Following the idea of \cite{EngelNagel2000}, we prove the existence of the
generator $\mathfrak{L}$ of a semigroup $\{\Lambda_t\}_{t\geq 0}$ on $V$.
The key observation is that, for a semigroup, strong continuity already forces
the existence of an infinitesimal generator. 
In the proof, we shall use the uniform-continuity formulation (mentioned in condition 4 above), only in order to keep the proof short and transparent. 
\begin{Thm}\label{thm:finite_dim_semigroup}
Let $\{\Lambda_t\}_{t\geq 0}$ be a uniform (or equivalently strongly) continuous one-parameter semigroup of linear maps on a finite-dimensional normed space $V$, with
$\Lambda_0=\I$. Then the limit
\[
    \mathfrak{L}v
    :=
    \lim_{t\downarrow 0}
    \frac{\Lambda_t v-v}{t}
\]
exists for every $v\in V$. The resulting linear map
$\mathfrak{L}:V\to V$ is called the generator of the semigroup
$\{\Lambda_t\}_{t\geq 0}$.
Moreover, $\Lambda_t$ is differentiable for every $t>0$ (and right-differentiable
at $t=0$) and satisfies
\[
    \frac{d}{dt}\Lambda_t v
    =
    \mathfrak{L}\Lambda_t v
    =
    \Lambda_t\mathfrak{L}v
    \qquad (t\geq 0,\ v\in V).
\]
Consequently,
\[
    \Lambda_t=\exp(t\mathfrak{L})
    \qquad (t\geq 0).
\]
\end{Thm}
\begin{proof}
By uniform continuity, the map $t\mapsto\Lambda_t$ is continuous from $[0,\infty)$ into the finite-dimensional normed space $\LA(V)$ equipped with the operator norm.

For $t\geq 0$, define
\[
    R(t):=\int_0^t \Lambda_s ds .
\]
(The integral may be understood componentwise after choosing a matrix representation of $\LA(V)$.)
Then $R$ is differentiable for every $t>0$, right-differentiable at $t=0$,
and
\[
    \frac{d}{dt}R(t)=\Lambda_t .
\]
Indeed, for $h\neq 0$ with $t+h\geq 0$,
\[
    \frac{R(t+h)-R(t)}{h}-\Lambda_t
    =
    \frac{1}{h}\int_t^{t+h}(\Lambda_s-\Lambda_t)\,ds .
\]
Hence
\[
    \left\|
    \frac{R(t+h)-R(t)}{h}-\Lambda_t
    \right\|_{\rm op}
    \leq
    \sup_{s\in I_{t,h}}\|\Lambda_s-\Lambda_t\|_{\rm op},
\]
where $I_{t,h}$ denotes the interval between $t$ and $t+h$. By uniform
continuity, the right-hand side tends to $0$ as $h\to 0$ whenever $t>0$,
and as $h\downarrow 0$ when $t=0$. Hence $R'(t)=\Lambda_t$.

In particular,
\[
    \frac{1}{t}R(t)
    =
    \frac{1}{t}\int_0^t\Lambda_s\,ds
    \longrightarrow
    \Lambda_0=\I
    \qquad (t\downarrow 0).
\]
Hence, for some sufficiently small $t_0>0$, $t_0^{-1}R(t_0)$ is sufficiently
close to $\I$, and is therefore invertible\footnote{
Intuitively, this follows from the fact that an operator sufficiently close to
the invertible operator $\I$ is again invertible. If one wants to justify this
point rigorously, one may argue as follows. Since $t^{-1}R(t)\to\I$ as $t\downarrow 0$, we can choose $t_0>0$
sufficiently small so that
\[
\|\I-t_0^{-1}R(t_0)\|_{\rm op}<1.
\]
Set $B:=\I-t_0^{-1}R(t_0)$. Then $\|B\|_{\rm op}<1$, and hence, by the Neumann series
(see Fact~\ref{ex:NeumannSeries} below), the inverse exists and is given by
\[
\sum_{n=0}^{\infty}B^n
=
(\I-B)^{-1}
=
(t_0^{-1}R(t_0))^{-1}.
\]
}. Hence $R(t_0)$ is also invertible.

Using the semigroup law, for every $t\geq 0$ we have
\[
    \Lambda_t R(t_0)
    =
    \int_0^{t_0}\Lambda_t\Lambda_s\,ds
    =
    \int_0^{t_0}\Lambda_{t+s}\,ds
    =
    \int_t^{t+t_0}\Lambda_r\,dr
    =
    R(t+t_0)-R(t).
\]
Since $R(t_0)$ is invertible, it follows that
\[
    \Lambda_t
    =
    \{R(t+t_0)-R(t)\}R(t_0)^{-1}.
\]
The right-hand side is differentiable for every $t>0$ and right-differentiable
at $t=0$. Hence the map $t\mapsto\Lambda_t$ is differentiable for every
$t>0$ and right-differentiable at $t=0$.

We now define
\[
    \mathfrak{L}
    :=
    \left.\frac{d}{dt}\Lambda_t\right|_{t=0+}.
\]
Equivalently, for each $v\in V$,
\[
    \mathfrak{L}v
    =
    \lim_{t\downarrow 0}
    \frac{\Lambda_t v-v}{t}.
\]
Thus the limit appearing in the statement exists for every $v\in V$.
Since the derivative at $t=0+$ is taken in the vector space $\LA(V)$, the
resulting map $\mathfrak{L}:V\to V$ is linear.

Next, for $t>0$, using the semigroup law and the differentiability already
established, we obtain
\[
    \frac{d}{dt}\Lambda_t
    =
    \lim_{h\downarrow 0}
    \frac{\Lambda_{t+h}-\Lambda_t}{h}
    =
    \lim_{h\downarrow 0}
    \frac{\Lambda_t\Lambda_h-\Lambda_t}{h}
    =
    \Lambda_t
    \lim_{h\downarrow 0}
    \frac{\Lambda_h-\I}{h}
    =
    \Lambda_t\mathfrak{L}.
\]
Similarly,
\[
    \frac{d}{dt}\Lambda_t
    =
    \lim_{h\downarrow 0}
    \frac{\Lambda_{h+t}-\Lambda_t}{h}
    =
    \lim_{h\downarrow 0}
    \frac{\Lambda_h\Lambda_t-\Lambda_t}{h}
    =
    \mathfrak{L}\Lambda_t.
\]
Therefore
\[
    \frac{d}{dt}\Lambda_t
    =
    \mathfrak{L}\Lambda_t
    =
    \Lambda_t\mathfrak{L}
    \qquad (t>0).
\]
At $t=0$, the same formula holds as a right derivative, because
$\Lambda_0=\I$.

Thus $\Lambda_t$ satisfies the linear differential equation
\[
    \frac{d}{dt}\Lambda_t
    =
    \mathfrak{L}\Lambda_t,
    \qquad
    \Lambda_0=\I.
\]
On the other hand, the operator exponential
\[
    \exp(t\mathfrak{L})
    :=
    \sum_{n=0}^{\infty}\frac{t^n\mathfrak{L}^n}{n!}
\]
satisfies the same differential equation with the same initial condition.
By uniqueness of solutions to linear ordinary differential equations in the
finite-dimensional space $\LA(V)$, we conclude that
\[
    \Lambda_t=\exp(t\mathfrak{L})
    \qquad (t\geq 0).
\]

Finally, the generator is unique. Indeed, if another linear map
$\mathfrak{L}'$ satisfies
\[
    \Lambda_t=\exp(t\mathfrak{L}')
    \qquad (t\geq 0),
\]
then differentiating at $t=0+$ gives
\[
    \mathfrak{L}'
    =
    \left.\frac{d}{dt}\exp(t\mathfrak{L}')\right|_{t=0+}
    =
    \left.\frac{d}{dt}\Lambda_t\right|_{t=0+}
    =
    \mathfrak{L}.
\]
This proves uniqueness.

For $v_t:=\Lambda_t v$, we also have
\[
    \frac{d}{dt}v_t
    =
    \frac{d}{dt}\Lambda_t v
    =
    \mathfrak{L}\Lambda_t v
    =
    \mathfrak{L}v_t.
\]
This completes the proof.
\end{proof}

\begin{F}[Neumann series]\label{ex:NeumannSeries}
Let $B$ be a linear map on a finite-dimensional normed vector space $V$.
If $\|B\|_{\rm op}<1$, then $\id-B$ is invertible and
\[
(\id-B)^{-1}
=
\sum_{n=0}^{\infty}B^n.
\]
\end{F}

\subsection{Proof of the Lumer--Phillips Theorem}\label{app:LP}

For completeness, we give a proof of the Lumer--Phillips theorem in the finite-dimensional setting.

\bigskip 
\noindent {\it Proof of Theorem \ref{thm:LP}.}
Suppose first that $\mathfrak{L}$ generates a contraction semigroup on $V$. Then
\[
\|e^{t\mathfrak{L}}v\|\leq \|v\|,
\qquad t\geq 0,
\]
for every $v\in V$. By the Cauchy--Schwarz property of the $L$-semi-inner product,
we have
\[
\operatorname{Re}\< v,e^{t\mathfrak{L}}v\>
\leq
|\< v,e^{t\mathfrak{L}}v\>|
\leq
\|v\|\|e^{t\mathfrak{L}}v\|
\leq
\|v\|^2.
\]
Since $\< v,v\>=\|v\|^2$, this implies
\[
\operatorname{Re}
\< v,e^{t\mathfrak{L}}v-v\>
\leq 0,
\qquad t\geq 0.
\]
By dividing by $t>0$ and using the linearity of the $L$-semi-inner product in the second argument, we obtain
\[
\operatorname{Re}
\< v,\frac{e^{t\mathfrak{L}}v-v}{t}\>
\leq 0.
\]
Note that, for fixed $v\in V$, the map $w\mapsto \< v,w\>$ is continuous, because $|\< v,w\>| \leq \|v\|\|w\|$. 
Therefore, taking the limit $t\downarrow 0$, we obtain
\[
\operatorname{Re}\< v,\mathfrak{L}v\>\leq 0.
\]
Thus $\mathfrak{L}$ is dissipative.

Conversely, suppose that $\mathfrak{L}$ is dissipative. 
For every $\lambda>0$ and every $v\in V$, using the linearity of the $L$-semi-inner product and the identity $\< v,v\>=\|v\|^2$, we obtain
\begin{align}
\operatorname{Re}\< v,(\lambda\I-\mathfrak{L})v\>
&=
\operatorname{Re}\left(\lambda\< v,v\>-\< v,\mathfrak{L}v\>\right) \\
&=
\lambda\|v\|^2
-
\operatorname{Re}\< v,\mathfrak{L}v\> \\
&\geq
\lambda\|v\|^2,
\end{align}
where the last inequality follows from the dissipativity of $\mathfrak{L}$.
By the Cauchy--Schwarz property again,
\[
\operatorname{Re}\< v,(\lambda\I-\mathfrak{L})v\>
\leq
|\< v,(\lambda\I-\mathfrak{L})v\>|
\leq
\|v\|\|(\lambda\I-\mathfrak{L})v\|.
\]
Hence, for $v\neq 0$, 
\[
\|(\lambda\I-\mathfrak{L})v\| \geq \lambda\|v\|.
\]
In particular, $\lambda\I-\mathfrak{L}$ is injective, i.e., $(\lambda\I-\mathfrak{L})v = 0$ implies $v = 0$. 
Since $V$ is finite-dimensional, $\lambda\I-\mathfrak{L}$ is invertible (see Fact \ref{InjInv}). Hence $(\lambda\I-\mathfrak{L})^{-1}$ is well-defined for every $\lambda>0$.
Let $w=(\lambda\I-\mathfrak{L})v$. Then $v=(\lambda\I-\mathfrak{L})^{-1}w$, and the above inequality gives
\[
\|w\|
\geq
\lambda\|(\lambda\I-\mathfrak{L})^{-1}w\|.
\]
Therefore,
\[
\|(\lambda\I-\mathfrak{L})^{-1}w\|
\leq
\frac{1}{\lambda}\|w\|.
\]
Taking the supremum over all $w\neq 0$, we obtain
\begin{equation}\label{ILm1}
\|(\lambda\I-\mathfrak{L})^{-1}\|_{\rm op}
\leq
\frac{1}{\lambda}.
\end{equation}
Now define
\[
\mathfrak{L}_\lambda
:=
\lambda\mathfrak{L}(\lambda\I-\mathfrak{L})^{-1},
\]
which is known as the Yosida approximation of $\mathfrak{L}$ \cite{EngelNagel2000}.

We first show that $e^{t\mathfrak{L}_\lambda}$ is contractive for every
$t\geq 0$. From
\[
\I
=
(\lambda\I-\mathfrak{L})(\lambda\I-\mathfrak{L})^{-1}
=
\lambda(\lambda\I-\mathfrak{L})^{-1}
-
\mathfrak{L}(\lambda\I-\mathfrak{L})^{-1},
\]
we obtain
\begin{equation}\label{identityL}
\mathfrak{L}(\lambda\I-\mathfrak{L})^{-1}
=
\lambda(\lambda\I-\mathfrak{L})^{-1}-\I.
\end{equation}
Therefore, 
\[
\mathfrak{L}_\lambda
=
\lambda^2(\lambda\I-\mathfrak{L})^{-1} - \lambda\I.
\]
Since $(\lambda\I-\mathfrak{L})^{-1}$ commutes with $\I$, we have
\[
e^{t\mathfrak{L}_\lambda}
= e^{t\lambda^2(\lambda\I-\mathfrak{L})^{-1} - t \lambda\I} = 
e^{-\lambda t}
e^{\lambda^2t(\lambda\I-\mathfrak{L})^{-1}}.
\]
Using the elementary estimate $\|e^A\|_{\rm op}\leq e^{\|A\|_{\rm op}}$ (see Fact \ref{expA}), and \eqref{ILm1}, we get
\begin{align}
\|e^{t\mathfrak{L}_\lambda}\|_{\rm op}
&\leq
e^{-\lambda t}
e^{\lambda^2t\|(\lambda\I-\mathfrak{L})^{-1}\|_{\rm op}} \\
&\leq
e^{-\lambda t}
e^{\lambda^2t\lambda^{-1}}
= 1.
\end{align}
Thus $e^{t\mathfrak{L}_\lambda}$ is contractive for every $t\geq 0$.

It remains to pass to the limit $\lambda\to\infty$. 
Using identity \eqref{identityL}, the submultiplicativity of the operator norm (see Fact \ref{ex:SMul}), and \eqref{ILm1}, we obtain
\[
\|\lambda(\lambda\I-\mathfrak{L})^{-1}-\I\|_{\rm op}
\leq
\|\mathfrak{L}\|_{\rm op}\|(\lambda\I-\mathfrak{L})^{-1}\|_{\rm op}
\leq
\frac{\|\mathfrak{L}\|_{\rm op}}{\lambda}.
\]
Therefore, $\lambda(\lambda\I-\mathfrak{L})^{-1} \to \I$ as $\lambda\to\infty$.
Since $\mathfrak{L}_\lambda = \mathfrak{L}\lambda(\lambda\I-\mathfrak{L})^{-1}$,
we obtain
\[
\mathfrak{L}_\lambda\to\mathfrak{L}.
\]
Finally, since $V$ is finite-dimensional, the exponential map is continuous.
Thus, for each $t\geq 0$,
\[
e^{t\mathfrak{L}_\lambda}
\to
e^{t\mathfrak{L}}.
\]
Consequently,
\[
\|e^{t\mathfrak{L}}\|_{\rm op}
=
\lim_{\lambda\to\infty}
\|e^{t\mathfrak{L}_\lambda}\|_{\rm op}
\leq 1.
\]
Hence $\mathfrak{L}$ generates a contraction semigroup.
\hfill $\square$ 

\begin{F}\label{InjInv}
Let $V$ be a finite-dimensional vector space and let $T:V\to V$ be a linear map.
If $T$ is injective, then $T$ is invertible.
\end{F}

\begin{F}\label{expA}
Let $A$ be a linear map on a finite-dimensional normed vector space $V$.
Then
\[
\|e^A\|_{\rm op} \leq e^{\|A\|_{\rm op}}.
\]
\end{F}

\section{Proofs of Facts}\label{app:Facts}

For the proof of Fact \ref{F:LExt}, we use the following elementary but useful lemma.

\begin{Lem}
Let $V$ be a real vector space, and let $S$ be a convex subset of $V$.
If $\Lambda:S\to S$ is an affine map, that is,
\[
\Lambda(ps_1+(1-p)s_2)
=
p\Lambda(s_1)+(1-p)\Lambda(s_2), \ (\forall s_1,s_2\in S, \forall p\in[0,1])
\]
then $\Lambda$ has a unique affine extension to the affine hull of $S$:
\[
{\rm aff}(S)
:=
\left\{
\sum_k \lambda_k s_k
\mid
s_k\in S,\ \sum_k\lambda_k=1
\right\}.
\]
\end{Lem}
\begin{proof}
For $x\in{\rm aff}(S)$, choose an affine representation
\[
x=\sum_k\lambda_k s_k,
\qquad
s_k\in S,
\qquad
\sum_k\lambda_k=1,
\]
and define
\[
\widetilde{\Lambda}(x)
:=
\sum_k\lambda_k\Lambda(s_k).
\]
We show that this definition is independent of the chosen affine representation.

Suppose that
\[
\sum_{i=1}^m\lambda_i s_i
=
\sum_{j=1}^n\mu_j t_j,
\qquad
s_i,t_j\in S,
\qquad
\sum_{i=1}^m\lambda_i
=
\sum_{j=1}^n\mu_j
=
1.
\]
It is enough to prove that
\[
\sum_{i=1}^m\lambda_i\Lambda(s_i)
=
\sum_{j=1}^n\mu_j\Lambda(t_j).
\]

Set
\[
\lambda_i^+:=\max\{\lambda_i,0\},
\qquad
\lambda_i^-:=\max\{-\lambda_i,0\},
\]
and define $\mu_j^+$ and $\mu_j^-$ similarly.
Then $\lambda_i=\lambda_i^+-\lambda_i^-$ and $\mu_j=\mu_j^+-\mu_j^-$ and $\lambda^\pm_i, \mu^\pm_j \ge 0$. 
From
\[
\sum_{i=1}^m\lambda_i s_i
=
\sum_{j=1}^n\mu_j t_j
\]
we obtain
\[
\sum_{i=1}^m\lambda_i^+ s_i
+
\sum_{j=1}^n\mu_j^- t_j
=
\sum_{i=1}^m\lambda_i^- s_i
+
\sum_{j=1}^n\mu_j^+ t_j.
\]
Moreover, since
\[
\sum_{i=1}^m\lambda_i
=
\sum_{j=1}^n\mu_j
=
1,
\]
the sums of the coefficients on the two sides are equal:
\[
c:= \sum_{i=1}^m\lambda_i^+
+
\sum_{j=1}^n\mu_j^-
=
\sum_{i=1}^m\lambda_i^-
+
\sum_{j=1}^n\mu_j^+.
\]
If $c=0$, then all coefficients above are zero, which is impossible because
$\sum_i\lambda_i=1$.
Thus $c>0$.
Dividing by $c$, we get an equality between two convex combinations of points of $S$:
\[
\sum_{i=1}^m\frac{\lambda_i^+}{c}s_i
+
\sum_{j=1}^n\frac{\mu_j^-}{c}t_j
=
\sum_{i=1}^m\frac{\lambda_i^-}{c}s_i
+
\sum_{j=1}^n\frac{\mu_j^+}{c}t_j.
\]
By the affinity of $\Lambda$ on $S$, this implies
\[
\sum_{i=1}^m\frac{\lambda_i^+}{c}\Lambda(s_i)
+
\sum_{j=1}^n\frac{\mu_j^-}{c}\Lambda(t_j)
=
\sum_{i=1}^m\frac{\lambda_i^-}{c}\Lambda(s_i)
+
\sum_{j=1}^n\frac{\mu_j^+}{c}\Lambda(t_j).
\]
Multiplying by $c$ and rearranging terms, we obtain
\[
\sum_{i=1}^m\lambda_i\Lambda(s_i)
=
\sum_{j=1}^n\mu_j\Lambda(t_j).
\]
Hence $\widetilde{\Lambda}$ is well defined.

It is clear that $\widetilde{\Lambda}$ extends $\Lambda$.
Moreover, if
\[
x=\sum_i\lambda_i s_i,
\qquad
y=\sum_j\mu_j t_j,
\qquad
\sum_i\lambda_i=\sum_j\mu_j=1,
\]
then, for $p\in[0,1]$,
\[
px+(1-p)y
=
\sum_i p\lambda_i s_i+\sum_j(1-p)\mu_j t_j,
\]
where the coefficients on the right-hand side sum to $1$.
Therefore
\[
\widetilde{\Lambda}(px+(1-p)y)
=
p\widetilde{\Lambda}(x)+(1-p)\widetilde{\Lambda}(y).
\]
Thus $\widetilde{\Lambda}$ is affine on ${\rm aff}(S)$.

Finally, if $\Gamma:{\rm aff}(S)\to{\rm aff}(S)$ is another affine extension of $\Lambda$, then for every
\[
x=\sum_k\lambda_k s_k,
\qquad
\sum_k\lambda_k=1,
\]
we have
\[
\Gamma(x)
=
\sum_k\lambda_k\Gamma(s_k)
=
\sum_k\lambda_k\Lambda(s_k)
=
\widetilde{\Lambda}(x).
\]
Hence $\Gamma=\widetilde{\Lambda}$.
\end{proof}

\noindent {\bf Fact \ref{F:LExt}.} {\it 
An affine map $\Lambda:\SA(\HA)\to\SA(\HA)$ uniquely extends to a trace-preserving positive real-linear map on $\LA_{\rm sa}(\HA)$. }

\bigskip 

\noindent [{\it Proof of Fact \ref{F:LExt}} ] Since the affine hull of $\SA(\HA)$ is
\[
{\rm aff}(\SA(\HA))
=
\{A\in\LA_{\rm sa}(\HA):\Tr A=1\},
\]
the preceding lemma shows that the affine map $\Lambda$ has a unique affine extension $\tilde{\Lambda}$ to this affine hull.

Fix any $\rho_0\in\SA(\HA)$.
For each $X\in\LA_{\rm sa}(\HA)$, write
\[
X=(\Tr X)\rho_0+Y,
\qquad
\Tr Y=0.
\]
Since $\rho_0+Y\in{\rm aff}(\SA(\HA))$, we define
\[
\widehat{\Lambda}(X)
:=
(\Tr X)\tilde{\Lambda}(\rho_0)
+
\tilde{\Lambda}(\rho_0+Y)-\tilde{\Lambda}(\rho_0).
\]
This is well defined, because the decomposition
$X=(\Tr X)\rho_0+Y$ with $\Tr Y=0$ is unique.

We first check real-linearity.
Let $X_i=t_i\rho_0+Y_i$ with $t_i=\Tr X_i,\ \Tr Y_i=0, \ i=1,2.$
Then
\[
X_1+X_2
=
(t_1+t_2)\rho_0+(Y_1+Y_2),
\qquad
\Tr(Y_1+Y_2)=0.
\]
By definition,
\[
\widehat{\Lambda}(X_1+X_2)
=
(t_1+t_2)\Lambda(\rho_0)
+
\tilde{\Lambda}(\rho_0+Y_1+Y_2)-\Lambda(\rho_0).
\]
Since
\[
\rho_0+Y_1+Y_2
=
1\times (\rho_0+Y_1)+ 1 \times (\rho_0+Y_2) + (-1) \times \rho_0
\]
is an affine combination of points in ${\rm aff}(\SA(\HA))$, the affinity of $\tilde{\Lambda}$ gives
\[
\tilde{\Lambda}(\rho_0+Y_1+Y_2)
=
\tilde{\Lambda}(\rho_0+Y_1)
+
\tilde{\Lambda}(\rho_0+Y_2)
-
\tilde{\Lambda}(\rho_0).
\]
Since $\tilde{\Lambda}(\rho_0)=\Lambda(\rho_0)$, we obtain
\[
\begin{aligned}
\widehat{\Lambda}(X_1+X_2)
&=
(t_1+t_2)\Lambda(\rho_0)
+
\tilde{\Lambda}(\rho_0+Y_1)
+
\tilde{\Lambda}(\rho_0+Y_2)
-
2\Lambda(\rho_0)\\
&=
\bigl(t_1\Lambda(\rho_0)+\tilde{\Lambda}(\rho_0+Y_1)-\Lambda(\rho_0)\bigr)
+
\bigl(t_2\Lambda(\rho_0)+\tilde{\Lambda}(\rho_0+Y_2)-\Lambda(\rho_0)\bigr)\\
&=
\widehat{\Lambda}(X_1)+\widehat{\Lambda}(X_2).
\end{aligned}
\]

Next, let $a\in\mathbb{R}$ and write $X=t\rho_0+Y$ where $t=\Tr X, \ \Tr Y=0.$
Then
\[
aX=(at)\rho_0+aY,
\qquad
\Tr(aY)=0.
\]
By definition,
\[
\widehat{\Lambda}(aX)
=
(at)\Lambda(\rho_0)
+
\tilde{\Lambda}(\rho_0+aY)-\Lambda(\rho_0).
\]
Since
\[
\rho_0+aY
=
a(\rho_0+Y)+(1-a)\rho_0
\]
is an affine combination of points in ${\rm aff}(\SA(\HA))$, the affinity of $\tilde{\Lambda}$ gives
\[
\tilde{\Lambda}(\rho_0+aY)
=
a\tilde{\Lambda}(\rho_0+Y)+(1-a)\tilde{\Lambda}(\rho_0).
\]
Using again $\tilde{\Lambda}(\rho_0)=\Lambda(\rho_0)$, we obtain
\[
\begin{aligned}
\widehat{\Lambda}(aX)
&=
at\Lambda(\rho_0)
+
a\tilde{\Lambda}(\rho_0+Y)
+
(1-a)\Lambda(\rho_0)
-
\Lambda(\rho_0)\\
&=
a\bigl(t\Lambda(\rho_0)+\tilde{\Lambda}(\rho_0+Y)-\Lambda(\rho_0)\bigr) =
a\widehat{\Lambda}(X).
\end{aligned}
\]
Thus $\widehat{\Lambda}$ is real-linear.

We also have trace preservation:
For $X=t\rho_0+Y$ with $t=\Tr X$ and $\Tr Y=0$, we have
\[
\Tr\widehat{\Lambda}(X)
=
t\Tr\Lambda(\rho_0)
+
\Tr\tilde{\Lambda}(\rho_0+Y)-\Tr\Lambda(\rho_0).
\]
Since $\tilde{\Lambda}$ maps the affine hull of the state space into itself, both
$\Lambda(\rho_0)$ and $\tilde{\Lambda}(\rho_0+Y)$ have trace $1$.
Therefore
\[
\Tr\widehat{\Lambda}(X)=t=\Tr X.
\]

Next, we show that $\widehat{\Lambda}$ is positive.
Let $X\ge 0$.
If $X=0$, then $\widehat{\Lambda}(X)=0$.
If $X\neq 0$, set $t=\Tr X>0$ and $\rho=X/t$.
Then $\rho\in\SA(\HA)$, and hence
\[
\widehat{\Lambda}(X)
=
t\widehat{\Lambda}(\rho)
=
t\Lambda(\rho)
\ge 0,
\]
because $\Lambda(\rho)\in\SA(\HA)$.

It remains to show uniqueness.
Let $\Phi:\LA_{\rm sa}(\HA)\to\LA_{\rm sa}(\HA)$ be any trace-preserving positive real-linear map extending $\Lambda$.
Then the restriction of $\Phi$ to ${\rm aff}\SA(\HA)$ is an affine extension of $\Lambda$.
By the uniqueness of the affine extension to ${\rm aff}\SA(\HA)$, we have
\[
\Phi(A)=\tilde{\Lambda}(A)
\]
for all $A\in{\rm aff}\SA(\HA)$.

Now let $X=t\rho_0+Y$ with $t=\Tr X$ and $\Tr Y=0$.
By real-linearity,
\[
\Phi(X)
=
t\Phi(\rho_0)+\Phi(Y).
\]
Moreover,
\[
\Phi(Y)
=
\Phi(\rho_0+Y)-\Phi(\rho_0)
=
\tilde{\Lambda}(\rho_0+Y)-\Lambda(\rho_0).
\]
Therefore
\[
\Phi(X)
=
t\Lambda(\rho_0)
+
\tilde{\Lambda}(\rho_0+Y)-\Lambda(\rho_0)
=
\widehat{\Lambda}(X).
\]
Thus the extension is unique. \hfill $\square$
\bigskip

\noindent {\bf Fact \ref{ex:TrAB}.}
{\it If $A,B$ are self-adjoint operators, then $\Tr AB\in\R$.
If $A,B$ are positive operators, then $\Tr AB\geq 0$. }

\bigskip 

\begin{proof}
Let $B=\sum_k b_k |b_k\rangle\langle b_k|$ be an eigenvalue decomposition of $B$.
Then
\[
\Tr AB
=
\sum_k b_k \Tr A|b_k\rangle\langle b_k|
=
\sum_k b_k \langle b_k|A|b_k\rangle.
\]
If $A$ and $B$ are self-adjoint, then $b_k\in\mathbb{R}$ and
\[
\langle b_k|A|b_k\rangle\in\mathbb{R}
\]
for every $k$.
Hence $\Tr AB\in\mathbb{R}$.

If, moreover, $A$ and $B$ are positive, then $b_k\ge 0$ and
\[
\langle b_k|A|b_k\rangle\ge 0
\]
for every $k$.
Therefore
\[
\Tr AB
=
\sum_k b_k \langle b_k|A|b_k\rangle
\ge 0.
\]
This proves the claim. 
\end{proof}

\bigskip 

\noindent {\bf Fact \ref{exe:TrnTrn}}. {\it 
For $A,B \in \LA(\HA)$, $|A\otimes B|=|A|\otimes |B|$.
In particular, taking the trace gives $\|A\otimes B\|_1=\|A\|_1\|B\|_1$.
}

\begin{proof}
Recalling $|X|:=(X^\dagger X)^{1/2}$, 
\[
|A\otimes B|^2
=
(A\otimes B)^\dagger(A\otimes B)
=
(A^\dagger\otimes B^\dagger)(A\otimes B)
=
A^\dagger A\otimes B^\dagger B = |A|^2\otimes |B|^2 = (|A|\otimes |B|)^2.
\]
Moreover, $|A|\otimes |B|$ is positive, because $|A|$ and $|B|$ are positive.
By the uniqueness of the positive square root, we obtain
\[
|A\otimes B|
=
|A|\otimes |B|.
\]
\end{proof}
\bigskip

\noindent {\bf Fact \ref{ex:TrAB=0}}. {\it 
If $\Tr(BA)=0$ for all $A \in \LA(\HA)$, then $B=0$.
}

\begin{proof}
For arbitrary $\psi,\phi\in\HA$, take $A=|\phi\rangle\langle \psi|$.
Then
\[
0=\Tr(BA)
=
\Tr B|\phi\rangle\langle \psi|
=
\langle \psi|B|\phi\rangle.
\]
Hence all matrix elements of $B$ vanish.
Therefore $B=0$.
\end{proof}

\bigskip
\noindent {\bf Fact \ref{fact:PartTen}}. {\it 
Let $\HA_1,\HA_2$ be Hilbert spaces of dimensions $d_1$ and $d_2$, respectively.
Let $\{\phi_i\}_{i=1}^{d_2}$ be an orthonormal basis of $\HA_2$.
If $\psi_i,\psi'_i\in\HA_1$, $i=1,\ldots,d_2$, satisfy
\[
\sum_{i=1}^{d_2}\psi_i\otimes\phi_i
=
\sum_{i=1}^{d_2}\psi'_i\otimes\phi_i,
\]
then $\psi_i=\psi'_i$ for all $i=1,\ldots,d_2$.
}

\begin{proof}
    Subtracting the two sides, we have
\[
\sum_{j=1}^{d_2}(\psi_j-\psi'_j)\otimes\phi_j=0.
\]
Applying $I_{\HA_1}\otimes\langle\phi_i|$ to both sides, we obtain
\[
0
=
\sum_{j=1}^{d_2}(\psi_j-\psi'_j)\langle\phi_i|\phi_j\rangle
=
\psi_i-\psi'_i.
\]
Hence $\psi_i=\psi'_i$ for all $i=1,\ldots,d_2$.
\end{proof}
\bigskip

\bigskip

\noindent {\bf Fact \ref{ex:NormEigChar}}. {\it 
For a self-adjoint operator $A \in \LA_{\rm sa}(\HA)$, 
\begin{equation}
\Tr A = \sum_{i=1}^d a_i, \ 
\|A\|_1 = \sum_{i=1}^d |a_i|, \ 
\|A\|_{\rm HS} = \sqrt{\sum_{i=1}^d a_i^2}, \ \|A\|_\infty = \max_{i=1,\ldots, d} |a_i|
\end{equation}
where $a_i \ (i=1,\ldots,d)$ are eigenvalues of $A$. 
}

\begin{proof}
Let $A=\sum_{i=1}^d a_i |\phi_i\rangle\langle\phi_i|$ be an eigenvalue decomposition of $A$. 
First, by the definition of the trace,
\[
\Tr A
=
\sum_{i=1}^d \langle \phi_i|A|\phi_i\rangle
=
\sum_{i=1}^d a_i.
\]

Next, since
\[
|A|
=
(A^\dagger A)^{1/2}
=
\sum_{i=1}^d |a_i| |\phi_i\rangle\langle\phi_i|,
\]
we have
\[
\|A\|_1
=
\Tr |A|
=
\sum_{i=1}^d |a_i|.
\]

Moreover,
\[
A^\dagger A
=
A^2
=
\sum_{i=1}^d a_i^2 |\phi_i\rangle\langle\phi_i|.
\]
Therefore
\[
\|A\|_{\rm HS}
=
\sqrt{\Tr A^\dagger A}
=
\sqrt{\sum_{i=1}^d a_i^2}.
\]

Finally, we prove the formula for the operator norm from the definition
\[
\|A\|_\infty:=\sup_{\|\psi\|=1}\|A\psi\|.
\]
Let $\psi\in\HA$ with $\|\psi\|=1$.
Write
\[
\psi=\sum_{i=1}^d c_i\phi_i \quad (\sum_{i=1}^d |c_i|^2=1).
\]
Then 
\[
A\psi=\sum_{i=1}^d a_i c_i\phi_i,
\]
and hence 
\[
\|A\psi\|^2
=
\sum_{i=1}^d |a_i|^2 |c_i|^2
\leq
\left(\max_{i=1,\ldots,d}|a_i|^2\right)
\sum_{i=1}^d |c_i|^2
=
\max_{i=1,\ldots,d}|a_i|^2.
\]
Therefore
\[
\|A\psi\|
\leq
\max_{i=1,\ldots,d}|a_i|.
\]
Taking the supremum over all unit vectors $\psi$, we get
\[
\|A\|_\infty
\leq
\max_{i=1,\ldots,d}|a_i|.
\]

Conversely, choose $i_0$ such that
\[
|a_{i_0}|=\max_{i=1,\ldots,d}|a_i|.
\]
Since $\|\phi_{i_0}\|=1$, we have
\[
\|A\|_\infty
\geq
\|A\phi_{i_0}\|
=
\|a_{i_0}\phi_{i_0}\|
=
|a_{i_0}|
=
\max_{i=1,\ldots,d}|a_i|.
\]
Thus
\[
\|A\|_\infty
=
\max_{i=1,\ldots,d}|a_i|.
\] 
\end{proof}

\bigskip

\noindent {\bf Fact \ref{XposYXY}}. {\it 
If $X$ is a positive operator, then $YXY^\dagger$ is also positive for any operator $Y$.
}
\begin{proof}
For any vector $\psi\in\HA$, we have
\[
\langle \psi|YXY^\dagger|\psi\rangle
=
\langle Y^\dagger\psi|X|Y^\dagger\psi\rangle.
\]
Since $X$ is positive, the right-hand side is nonnegative.
Hence $YXY^\dagger$ is positive.
\end{proof}
\bigskip
\noindent {\bf Fact \ref{f:TrAB}}. {\it 
If $\Tr XA = \Tr XB$ for all $X \in \LA(\HA)$, then $A = B$. 
}
\begin{proof}
By linearity of the trace and its cyclicity,
\[
0=\Tr X(A-B)=\Tr (A-B)X
\]
for all $X\in\LA(\HA)$.
Hence, by Fact \ref{ex:TrAB=0}, we have $A-B = 0$.
\end{proof}

\bigskip
\noindent {\bf Fact \ref{ex:SMul}}. {\it 
Let $V$ be a normed vector space.
For linear maps $\Lambda$ and $\Gamma$ on $V$, the operator norm is submultiplicative:
\[
\|\Lambda\Gamma\|_{\rm op}
\leq
\|\Lambda\|_{\rm op}\|\Gamma\|_{\rm op}.
\]
Moreover, if $\I$ denotes the identity map on $V$, then
\[
\|\I\|_{\rm op}=1.
\]
}
\begin{proof}
Recall that
\[
\|\Lambda\|_{\rm op}
:=
\sup_{\|x\|=1}\|\Lambda x\|.
\]
For any $x\in V$ with $\|x\|=1$, we have
\[
\|\Lambda\Gamma x\|
\leq
\|\Lambda\|_{\rm op}\|\Gamma x\|.
\]
Moreover,
\[
\|\Gamma x\|
\leq
\|\Gamma\|_{\rm op}\|x\|
=
\|\Gamma\|_{\rm op}.
\]
Therefore
\[
\|\Lambda\Gamma x\|
\leq
\|\Lambda\|_{\rm op}\|\Gamma\|_{\rm op}.
\]
Taking the supremum over all $x\in V$ with $\|x\|=1$, we obtain
\[
\|\Lambda\Gamma\|_{\rm op}
\leq
\|\Lambda\|_{\rm op}\|\Gamma\|_{\rm op}.
\]

Next, for any $x\in V$ with $\|x\|=1$, we have
\[
\|\I x\|=\|x\|=1.
\]
Hence
\[
\|\I\|_{\rm op}
=
\sup_{\|x\|=1}\|\I x\|
=
1.
\]
\end{proof}
\bigskip
\noindent {\bf Fact \ref{ex:NeumannSeries}}. {\it 
Let $B$ be a linear map on a finite-dimensional normed vector space $V$.
If $\|B\|_{\rm op}<1$, then $\id-B$ is invertible and
\[
(\id-B)^{-1}
=
\sum_{n=0}^{\infty}B^n.
\]
}

\begin{proof}
Since $\|B\|_{\rm op}<1$, the series
\[
S: = \sum_{n=0}^{\infty}B^n
\]
converges absolutely with respect to the operator norm.
Indeed, using the submultiplicativity of the operator norm, $\|AB\|_{\rm op}\leq \|A\|_{\rm op}\|B\|_{\rm op}$, we have
\[
\sum_{n=0}^{\infty}\|B^n\|_{\rm op}
\leq
\sum_{n=0}^{\infty}\|B\|_{\rm op}^n
=
\frac{1}{1-\|B\|_{\rm op}}
<
\infty.
\]
By continuity of composition with respect to the operator norm,
\[
(\id-B)S
=
\sum_{n=0}^{\infty}B^n-\sum_{n=0}^{\infty}B^{n+1}
=
\id.
\]
Similarly,
\[
S(\id-B)
=
\sum_{n=0}^{\infty}B^n-\sum_{n=0}^{\infty}B^{n+1}
=
\id.
\]
Therefore, $\id-B$ is invertible and
\[
(\id-B)^{-1}
=
\sum_{n=0}^{\infty}B^n.
\]
\end{proof}

\bigskip
\noindent {\bf Fact \ref{InjInv}}. {\it 
Let $V$ be a finite-dimensional vector space and let $T:V\to V$ be a linear map.
If $T$ is injective, then $T$ is invertible.
}

\begin{proof}Choose a basis of $V$ and let $A$ be the matrix representation of $T$ with respect to this basis.
Since $T$ is injective, the equation $Ax=0$ has only the trivial solution.
Hence the columns of $A$ are linearly independent.
By the elementary property of determinants, this implies $\det A\neq 0$, and thus $A$ is invertible. Hence $T$ is invertible.
\end{proof}

\bigskip
\noindent {\bf Fact \ref{expA}}. {\it 
Let $A$ be a linear map on a finite-dimensional normed vector space $V$.
Then
\[
\|e^A\|_{\rm op} \leq e^{\|A\|_{\rm op}}.
\]
}
\begin{proof}
By definition,
\[
e^A
=
\sum_{n=0}^{\infty}\frac{A^n}{n!}.
\]
Using the triangle inequality and the submultiplicativity of the operator norm, we obtain
\[
\|e^A\|_{\rm op}
=
\left\|
\sum_{n=0}^{\infty}\frac{A^n}{n!}
\right\|_{\rm op}
\leq
\sum_{n=0}^{\infty}\frac{\|A^n\|_{\rm op}}{n!}
\leq
\sum_{n=0}^{\infty}\frac{\|A\|_{\rm op}^n}{n!}
=
e^{\|A\|_{\rm op}}.
\]
This proves the assertion.
\end{proof}
\bibliographystyle{unsrt}
\bibliography{ref}

@article{GoriniKossakowskiSudarshan1976,
  author  = {Gorini, Vittorio and Kossakowski, Andrzej and Sudarshan, E. C. G.},
  title   = {Completely positive dynamical semigroups of {$N$}-level systems},
  journal = {Journal of Mathematical Physics},
  volume  = {17},
  number  = {5},
  pages   = {821--825},
  year    = {1976},
  doi     = {10.1063/1.522979}
}

@article{Lindblad1976,
  author  = {Lindblad, G{\"o}ran},
  title   = {On the generators of quantum dynamical semigroups},
  journal = {Communications in Mathematical Physics},
  volume  = {48},
  number  = {2},
  pages   = {119--130},
  year    = {1976},
  doi     = {10.1007/BF01608499}
}

@article{Kossakowski1972a,
  author  = {Kossakowski, Andrzej},
  title   = {On quantum statistical mechanics of non-Hamiltonian systems},
  journal = {Reports on Mathematical Physics},
  volume  = {3},
  number  = {4},
  pages   = {247--274},
  year    = {1972},
  doi     = {10.1016/0034-4877(72)90010-9}
}

@article{Kossakowski1972b,
  author  = {Kossakowski, Andrzej},
  title   = {On necessary and sufficient conditions for a generator of a quantum dynamical semi-group},
  journal = {Bulletin de l'Acad{\'e}mie Polonaise des Sciences, S{\'e}rie des Sciences Math{\'e}matiques, Astronomiques et Physiques},
  volume  = {20},
  pages   = {1021--1025},
  year    = {1972}
}

@article{Kossakowski1973,
  author  = {Kossakowski, Andrzej},
  title   = {On the general form of the generator of a dynamical semi-group for the spin $1/2$ system},
  journal = {Bulletin de l'Acad{\'e}mie Polonaise des Sciences, S{\'e}rie des Sciences Math{\'e}matiques, Astronomiques et Physiques},
  volume  = {21},
  pages   = {649--653},
  year    = {1973}
}

@book{RivasHuelga2012,
  author    = {Rivas, {\'A}ngel and Huelga, Susana F.},
  title     = {Open Quantum Systems: An Introduction},
  series    = {SpringerBriefs in Physics},
  publisher = {Springer},
  address   = {Berlin, Heidelberg},
  year      = {2012},
  doi       = {10.1007/978-3-642-23354-8},
  isbn      = {978-3-642-23353-1}
}

@book{EngelNagel2000,
  author    = {Engel, Klaus-Jochen and Nagel, Rainer},
  title     = {One-Parameter Semigroups for Linear Evolution Equations},
  series    = {Graduate Texts in Mathematics},
  volume    = {194},
  publisher = {Springer},
  address   = {New York},
  year      = {2000},
  isbn      = {978-0-387-98463-6}
}

@book{NC,
  author    = {Nielsen, Michael A. and Chuang, Isaac L.},
  title     = {Quantum Computation and Quantum Information},
  edition   = {10th Anniversary Edition},
  publisher = {Cambridge University Press},
  address   = {Cambridge},
  year      = {2010}
}

@book{OurTextbook,
  author    = {Hayashi, Masahito and Ishizaka, Satoshi and Kawachi, Akinori and Kimura, Gen and Ogawa, Tomohiro},
  title     = {Introduction to Quantum Information Science},
  series    = {Graduate Texts in Physics},
  publisher = {Springer},
  address   = {Berlin, Heidelberg},
  year      = {2015},
  doi       = {10.1007/978-3-662-43502-1},
  isbn      = {978-3-662-43502-1}
}

@book{Wilde,
  author    = {Wilde, Mark M.},
  title     = {Quantum Information Theory},
  publisher = {Cambridge University Press},
  address   = {Cambridge},
  edition   = {Second},
  year      = {2017}
}

@book{Watrous,
  author    = {Watrous, John},
  title     = {The Theory of Quantum Information},
  publisher = {Cambridge University Press},
  address   = {Cambridge},
  year      = {2018}
}

@book{HeinosaariZiman,
  author    = {Heinosaari, Teiko and Ziman, M{\'a}rio},
  title     = {The Mathematical Language of Quantum Theory: From Uncertainty to Entanglement},
  publisher = {Cambridge University Press},
  address   = {Cambridge},
  year      = {2012}
}

@article{LumerPhillips1961,
  author  = {Lumer, G{\"u}nter and Phillips, Ralph S.},
  title   = {Dissipative operators in a Banach space},
  journal = {Pacific Journal of Mathematics},
  volume  = {11},
  number  = {2},
  pages   = {679--698},
  year    = {1961}
}

@article{Chruscinski2021LegacyKossakowski,
  author  = {Chru{\'s}ci{\'n}ski, Dariusz},
  title   = {The Legacy of {A}ndrzej {K}ossakowski},
  journal = {Open Systems \& Information Dynamics},
  volume  = {28},
  number  = {4},
  pages   = {2150015},
  year    = {2021},
  doi     = {10.1142/S1230161221500153}
}

@article{Choi1975,
  author  = {Choi, Man-Duen},
  title   = {Completely positive linear maps on complex matrices},
  journal = {Linear Algebra and its Applications},
  volume  = {10},
  number  = {3},
  pages   = {285--290},
  year    = {1975},
  doi     = {10.1016/0024-3795(75)90075-0}
}

@article{Jamiolkowski1972,
  author  = {Jamio{\l}kowski, Andrzej},
  title   = {Linear transformations which preserve trace and positive semidefiniteness of operators},
  journal = {Reports on Mathematical Physics},
  volume  = {3},
  number  = {4},
  pages   = {275--278},
  year    = {1972},
  doi     = {10.1016/0034-4877(72)90011-0}
}

@article{Kraus1971,
  author  = {Kraus, Karl},
  title   = {General state changes in quantum theory},
  journal = {Annals of Physics},
  volume  = {64},
  number  = {2},
  pages   = {311--335},
  year    = {1971},
  doi     = {10.1016/0003-4916(71)90108-4}
}

@article{Stinespring1955,
  author  = {Stinespring, W. Forrest},
  title   = {Positive functions on {$C^*$}-algebras},
  journal = {Proceedings of the American Mathematical Society},
  volume  = {6},
  number  = {2},
  pages   = {211--216},
  year    = {1955},
  doi     = {10.1090/S0002-9939-1955-0069403-4}
}

@book{RudinPMA,
  author    = {Rudin, Walter},
  title     = {Principles of Mathematical Analysis},
  edition   = {3},
  publisher = {McGraw-Hill},
  address   = {New York},
  year      = {1976}
}

@book{Kreyszig,
  author    = {Kreyszig, Erwin},
  title     = {Introductory Functional Analysis with Applications},
  publisher = {John Wiley \& Sons},
  address   = {New York},
  year      = {1978}
}

@book{BreuerPetruccione2002,
  author    = {Breuer, Heinz-Peter and Petruccione, Francesco},
  title     = {The Theory of Open Quantum Systems},
  publisher = {Oxford University Press},
  address   = {Oxford},
  year      = {2002}
}

@book{GardinerZoller2004,
  author    = {Gardiner, Crispin and Zoller, Peter},
  title     = {Quantum Noise: A Handbook of Markovian and Non-Markovian Quantum Stochastic Methods with Applications to Quantum Optics},
  edition   = {3},
  publisher = {Springer},
  address   = {Berlin},
  year      = {2004},
  isbn      = {978-3-540-22301-6}
}

@book{AlickiLendi2007,
  author    = {Alicki, Robert and Lendi, Karl},
  title     = {Quantum Dynamical Semigroups and Applications},
  series    = {Lecture Notes in Physics},
  volume    = {717},
  publisher = {Springer},
  address   = {Berlin},
  year      = {2007},
  doi       = {10.1007/3-540-70861-8}
}

@article{Lumer1961,
  author  = {Lumer, G{\"u}nter},
  title   = {Semi-inner product spaces},
  journal = {Transactions of the American Mathematical Society},
  volume  = {100},
  pages   = {29--43},
  year    = {1961}
}

@misc{WikipediaKossakowski,
  author       = {{Wikipedia contributors}},
  title        = {{Andrzej Kossakowski}},
  howpublished = {\url{https://en.wikipedia.org/w/index.php?title=Andrzej_Kossakowski&oldid=1263442218}},
  note         = {Wikipedia, The Free Encyclopedia. Accessed May 15, 2026},
  year         = {2024}
}

@article{ChruscinskiPascazio2017,
  author  = {Chru{\'s}ci{\'n}ski, Dariusz and Pascazio, Saverio},
  title   = {A Brief History of the {GKLS} Equation},
  journal = {Open Systems \& Information Dynamics},
  volume  = {24},
  number  = {03},
  pages   = {1740001},
  year    = {2017},
  doi     = {10.1142/S1230161217400017},
  eprint  = {1710.05993},
  archivePrefix = {arXiv},
  primaryClass  = {quant-ph}
}

@article{Saverio,
author = {Pascazio, Saverio},
title = {Remembering {A}ndrzej {K}ossakowski},
journal = {Open Systems \& Information Dynamics},
volume = {28},
number = {04},
pages = {2150016},
year = {2021},
doi = {10.1142/S1230161221500165},
URL = {https://doi.org/10.1142/S1230161221500165},
eprint = {https://doi.org/10.1142/S1230161221500165}
}

@article{Rel1,
  author  = {Kimura, Gen},
  title   = {Restriction on relaxation times derived from the {Lindblad}-type master equations for two-level systems},
  journal = {Physical Review A},
  volume  = {66},
  pages   = {062113},
  year    = {2002},
  doi     = {10.1103/PhysRevA.66.062113}
}

@article{Rel2,
  author  = {Chru{\'s}ci{\'n}ski, Dariusz and Kimura, Gen and Kossakowski, Andrzej and Shishido, Yasuhito},
  title   = {Universal Constraint for Relaxation Rates for Quantum Dynamical Semigroup},
  journal = {Physical Review Letters},
  volume  = {127},
  pages   = {050401},
  year    = {2021},
  doi     = {10.1103/PhysRevLett.127.050401}
}

@article{Rel3,
  author  = {Muratore-Ginanneschi, Paolo and Kimura, Gen and Chru{\'s}ci{\'n}ski, Dariusz},
  title   = {Universal bound on the relaxation rates for quantum {Markovian} dynamics},
  journal = {Journal of Physics A: Mathematical and Theoretical},
  volume  = {58},
  pages   = {045306},
  year    = {2025},
  doi     = {10.1088/1751-8121/adaa3f}
}

@article{Rel4,
  author  = {Chru{\'s}ci{\'n}ski, Dariusz and vom Ende, Frederik and Kimura, Gen and Muratore-Ginanneschi, Paolo},
  title   = {A universal constraint for relaxation rates for quantum {Markov} generators: complete positivity and beyond},
  journal = {Reports on Progress in Physics},
  volume  = {88},
  pages   = {097602},
  year    = {2025},
  doi     = {10.1088/1361-6633/ae075f}
}
\end{document}